\documentclass[runningheads]{llncs}
\usepackage[T1]{fontenc}
\usepackage{amsmath}
\usepackage{multicol}
\usepackage{amsfonts}
\usepackage{xcolor}
\usepackage{graphicx}
\usepackage{tcolorbox}
\usepackage{float}

\usepackage{tcolorbox}
\usepackage{tikz}
\usepackage[linesnumbered, ruled, vlined]{algorithm2e}
\usepackage{caption}
\usepackage{subcaption}
\usepackage{cite}
\usepackage{caption}
\usepackage{url}
\usepackage{hyperref}
\hypersetup{
    colorlinks=true,
    linkcolor=blue,
    filecolor=magenta,      
    urlcolor=cyan,
}
\begin{document}
\title{Influential Slot and Tag Selection in Billboard Advertisement}
\author{Dildar Ali \and Suman Banerjee \and Yamuna Prasad}
\authorrunning{Ali et al.} % abbreviated author list (for running head)
\institute{Indian Institute of Technology Jammu, J \& K-181221, India \email{\{2021rcs2009,suman.banerjee,yamuna.prasad\}@iitjammu.ac.in}}
\maketitle

\begin{abstract}
The selection of influential billboard slots remains an important problem in billboard advertisements. Existing studies on this problem have not considered the case of context-specific influence probability. To bridge this gap, in this paper, we introduce the \textsc{Context Dependent Influential Billboard Slot Selection Problem}. First, we show that the problem is NP-hard. We also show that the influence function holds the bi-monotonicity, bi-submodularity, and non-negativity properties. We propose an orthant-wise Stochastic Greedy approach to solve this problem. We show that this method leads to a constant-factor approximation guarantee. Subsequently, we propose an orthant-wise Incremental and Lazy Greedy approach. In a generic sense, this is a method for maximizing a bi-submodular function under the cardinality constraint, which may also be of independent interest. We analyze the performance guarantee of this algorithm as well as time and space complexity. The proposed solution approaches have been implemented with real-world billboard and trajectory datasets. We compare the performance of our method with several baseline methods, and the results are reported. Our proposed orthant-wise stochastic greedy approach leads to significant results when the parameters are set properly with reasonable computational overhead.
\keywords{Billboard Advertisement \and Billboard Database \and Trajectory Database \and Influence Probability \and Bi-submodularity}
\end{abstract}

\section{Introduction} 
In recent times, \emph{Billboard Advertisement} has emerged as an effective out-of-home advertisement technique due to multiple reasons such as being easy to adopt, ensuring a return on investment\footnote{\url{https://www.thebusinessresearchcompany.com/report/billboard-and-outdoor-advertising-global-market-report}}. If we have the location information of a group of people over different time stamps and locations of a set of billboards then appropriate advertisement contents could be displayed on the billboards, and it may lead to an influence among the people. In billboard advertisements, the billboards are owned by some billboard owners (e.g., Lamar, Sigtel, etc.), and different commercial houses approach a billboard owner for a number of billboard slots depending on their budget. Given a trajectory database, a billboard database, and a positive integer $k$, which $k$ billboard slots should be chosen to maximize the influence? This problem has been referred to as the \textsc{Top-$k$ Influential Billboard Slot Selection Problem} \cite{ali2023influential}, and a few solution methodologies are available. The influence probability between a billboard slot and a trajectory has been considered in all these studies \cite{ali2023influential,10.1145/3292500.3330829,ali2025multi} is the same and does not vary. However, in practice, a low-income person will be more influenced toward a low-cost product rather than a high-cost product. Hence, the influence probability is dependent on context, and this notion is captured as a tag-dependent influence probability. In recent times, tag-based influence maximization has gained significant attention, and most of the studies on this topic are concerned with social networks \cite{banerjee2020budgeted,tekawade2023influence}. In both studies, the authors have proposed a bi-set function and an incremental greedy approach that exploits the submodularity property of the influence function. However, such studies have not been done in the context of billboard advertisement, although actual influence is dependent on both slot and tag. Now, the question is that given two positive integers $k$ and $\ell$, which $k$ influential slots and $\ell$ influential tags should be chosen such that the influence is maximized. To the best of our knowledge, such a problem has not been addressed in billboard advertisement settings. However, some studies focus on the influence maximization in the presence of tags in social networks. The first study by Ke et al. \cite{DBLP:conf/sigmod/KeKC18}, where they studied the problem of finding $k$ seed nodes and $r$ influential tags to maximize the influence in the network. Subsequently, other solution methodologies exist in the literature, e.g., the community-based approach \cite{banerjee2022budgeted} that exploits the bi-submodularity of the influence function. This paper bridges this gap by studying the influential billboard slot selection problem in tag-specific influence probability settings. We have posed this problem as a maximization of the bi-submodular set function \cite{schoot2024characterization}. In the literature, several practical problems have been modeled as a maximization of a bi-submodular function, such as influence maximization in social networks \cite{tekawade2023influence}, drug-drug interaction detection \cite{hu2017drug}, and many more. In particular, we make the following contributions in this paper:
\begin{itemize}
\item We study this problem in the tag-specific influence probability setting, where the goal is to select influential slots and tags to maximize the influence.
\item We establish several important properties of the influence function and exploit them to design efficient algorithms to solve this problem. 
\item We propose an efficient Orthent-wise Stochastic Greedy maximization algorithm and subsequently introduce Incremental and Lazy Greedy algorithms.
\item We analyze the algorithm to understand its time and space complexities, performance guarantee, and conduct experiments with real-world trajectory datasets to exhibit the effectiveness and efficiency of the proposed approach.     
\end{itemize}
\par Rest of the paper is organized as follows. Section \ref{Sec:BPD} describes the required background and defines our problem formally. The proposed solution approaches have been described in Section \ref{Sec:Proposed}. Section \ref{Sec:EE} describes the experimental evaluations of the proposed solutions. Finally, Section \ref{Sec:Conclusion} concludes this study and provides future research directions.

\section{Background and Problem Definition} \label{Sec:BPD}
\subsection{Trajectory and Billboard and Tag Database}
A trajectory database contains location information of moving objects over time. In this problem context, the trajectory database $\mathcal{D}$ contains tuples of the form $(\mathcal{U}^{'}, \texttt{loc}, [t_1,t_2])$, signifies the set of people $\mathcal{U}^{'}$ was at the location $\texttt{loc}$ for the duration $[t_1,t_2]$. For any tuple $p \in \mathcal{D}$, let $p_{u}$ denote the set of people associated with it. Let $\mathcal{U}=\{u_1, u_2, \ldots, u_n\}$ denote the set of people covered by the trajectory database, and hence $\mathcal{U}=\underset{p \in \mathcal{D}}{\bigcup} p_{u}$. This is defined as the people for which there exists at least one tuple that contains the people, i.e., $\mathcal{U}=\{u_i: \exists (\mathcal{U}^{'}, \texttt{loc}, [t_1,t_2]) \in \mathcal{D} \texttt{ and } u_i \in \mathcal{U}^{'} \}$. Similarly, $\mathcal{L}$ denotes the set of locations that are covered by the trajectory database $\mathcal{D}$, i.e., $\mathcal{L}=\{\texttt{loc}_{i}: (\mathcal{U}^{'}, \texttt{loc}_{i}, [t_j,t_k]) \in \mathcal{D} \}$. Let $[T_1, T_2]$ be the duration for which the trajectory database $\mathcal{D}$ contains the movement data. The billboard database $\mathcal{B}$ stores information about billboards across a city. Each entry is a tuple $(b_{id}, \texttt{loc}, \texttt{slot\_duration}, \texttt{cost})$, where $b_{id}$ is the billboard ID, $\texttt{loc}$ is the location, $\texttt{slot\_duration}$ is the slot duration, and $\texttt{cost}$ is the associated cost. Assume all billboards operate over the period $[T_1, T_2]$, with each slot having duration $\Delta$. A billboard slot is represented as a tuple of billboard ID and slot duration. The set of all billboard slots, denoted as $\mathcal{BS}$, is defined as: $\mathcal{BS} = \{(b_i, [t_j, t_j + \Delta]) : i \in [m] \text{ and } t_j \in \{1, \Delta+1, 2\Delta+1, \ldots, T_2-\Delta+1\}\}$. The tag database contains information about tags (i.e., advertisement content) from the commercial clients. The tag database $\mathcal{T}$ contains a tuple of the form $(tag\_id, tag\_cost)$, which signifies each tag contains its corresponding unique tag ID and cost. The influence providers allocate slots to commercial clients to maximize product influence. The key question becomes: how can we quantify the influence of a set of billboard slots? This is addressed in Definition \ref{Def:Influence}.

\begin{definition} [Influence of Billboard Slots] \label{Def:Influence}
Given a trajectory database $\mathcal{D}$, and a subset of billboard slots $\mathcal{S} \subseteq \mathcal{BS}$, the influence of $\mathcal{S}$ can be defined as the expected number of trajectories is influenced can be computed using Equation \ref{Eq:Equation_Inf}.

\begin{equation} \label{Eq:Equation_Inf}
    \phi(\mathcal{S})= \underset{u_i \in \mathcal{U}}{\sum} [1- \underset{b_j \in \mathcal{S}}{\prod} (1-Pr(b_j,u_i))]
\end{equation}
\end{definition}

Here, $\phi$ is the influence function that maps each subset of the billboard slots to its expected influence, hence $\phi: 2^{\mathcal{BS}} \longrightarrow \mathbb{R}_{0}^{+}$ and $\phi(\emptyset)=0$. The influence model stated in Definition \ref{Def:Influence} has been widely accepted in the existing studies \cite{zhang2020towards,ali2023influential} on billboard advertisement. Assuming that a person $u_i$ crosses a billboard slot, $bs_i$, at a time $t_x$. Now, assume that the advertisement content of an E-Commerce house is displayed on that billboard in the slot $[t_i,t_j]$, and $t_x \in [t_i,t_j]$. Then $u_i$ is likely to be influenced by the advertisement content with a certain probability. The billboard $b_i$ will influence the user $u_i$ with probability $Pr(bs_i,u_i)$. One of the way to calculate this value as,  $Pr(bs_i,u_i) = \frac{Size(bs_i)}{\underset{bs_{i} \in \mathcal{BS}}{max} \ Size(bs_{i})} $ where $Size(bs_{i})$ is the billboard panel size. We adopt this probability setting in our experiments as well. Although it can be calculated in several ways depending on the needs of applications \cite {10.1145/3292500.3330829,8604082,zhang2020towards}. As mentioned in the literature \cite{DBLP:conf/sigmod/KeKC18}, whether a people will be influenced towards a brand or not is always context dependent. In this study we consider that every people $u_i \in \mathcal{U}$, for a billboard slot $b_j \in \mathcal{BS}$ and every relevant tag $c \in \mathcal{H}^{'}$, there exists a non-zero tag specific influence probability and it is denoted by $Pr(u_i,b_j|c)$. This signifies the influence probability of the people $u_i$ when he/she looks at some advertisement content containing the tag $\textbf{c}$ at the billboard slot $b_j$.  For any person, $u \in \mathcal{U}$, for any set of given tags $\mathcal{H}^{'} \subseteq \mathcal{H}$, it is an important question how to calculate the aggregated influence of the tags in $\mathcal{H}^{'}$. Assume that $\mathcal{H}^{'}(u)$ denotes the subset of the tags used in the advertisement content which are visible to $u$. Now, the aggregated influence will be dependent on how the tags are aggregated. In this study, we use the independent tag aggregation, which has been stated in Definition \ref{Def:tag_aggregation}.

\begin{definition}[Independent Tag Aggregation] \label{Def:tag_aggregation}
For a subset of given slots $\mathcal{S} \subseteq \mathcal{BS}$, tags  $\mathcal{H}^{'}$, as per independent tag aggregation the aggregated influence probability of $u$ can be computed using Equation \ref{Eq:Inf_1}.

\begin{equation} \label{Eq:Inf_1}
    Pr(u, \mathcal{S}|\mathcal{H}^{'})=1-\underset{(b,c) \in f}{\prod} \ (1-Pr(u,b|c))
\end{equation}
\end{definition}
Here, $f$ denotes the tag assignment function.  Now, the aggregated influence for a given subset of billboard slots $\mathcal{S}$, a set of given tags $\mathcal{H}^{'}$ is denoted by $\Phi(\mathcal{S}, \mathcal{H}^{'})$ and stated in Definition \ref{Def:Aggregated_Influence}.

\begin{definition}[Aggregated Influence] \label{Def:Aggregated_Influence}
The aggregated influence for a given subset of billboard slots $\mathcal{S}$, a set of given tags $\mathcal{H}^{'}$ is defined as the sum of the influence probabilities of all the persons as stated in Equation \ref{Eq:Aggregated_Influence}. 

\begin{equation}\label{Eq:Aggregated_Influence}
    \Phi(\mathcal{S}, \mathcal{H}^{'})= \underset{u \in \mathcal{U}}{\sum} \ Pr(u|\mathcal{H}^{'})
\end{equation}
Here, $\Phi(.,.)$ is a bi-set function which is a mapping from  $2^{\mathcal{BS}} \times 2^{\mathcal{H}^{'}}$ to the set of positive real number including $0$, i.e.,  $\Phi: 2^{\mathcal{BS}} \times 2^{\mathcal{H}^{'}} \longrightarrow \mathbb{R}_{0}^{+}$. It can be observed that for any subset of slots $\mathcal{S}$, if the tag set is $\emptyset$ then the aggregated influence will be $0$. Hence, $\Phi(\mathcal{S}, \emptyset)=0$ for all $\mathcal{S} \subseteq \mathcal{BS}$.
\end{definition} 
\subsection{Problem Definition}
As previously mentioned, selecting both tags and billboard slots is important. However, obtaining a required billboard slot from an influence provider is subject to payment, and the e-commerce house doing this advertisement will have budget constraints. So, the goal here is to select $k$ many billboard slots and $\ell$ many tags to maximize the influence. We call this problem the \textsc{Context Dependent Influential Billboard Slot Selection Problem}, which asks for given $k$ and $\ell$, in which $k$ influential slots and tags should be chosen respectively to maximize the influence. We state the problem in Definition \ref{Def:Problem1}.

\begin{definition}[Context Dependent Influential
Billboard Slot Selection Problem] \label{Def:Problem1}
Given a trajectory database $\mathcal{D}$, a billboard database $\mathcal{B}$, and two positive integers $k$ and $\ell$, this problem asks to choose $k$ influential billboard slots and $\ell$ influential tags such that the influence is maximized. Mathematically, this problem can be expressed as follows:
\begin{equation} \label{Eq:4}
    (\mathcal{S}^{OPT}, \mathcal{H}^{OPT}) \longleftarrow \underset{\mathcal{S} \subseteq \mathcal{BS} \wedge \mathcal{H}^{'} \subseteq \mathcal{H} }{argmax} \ \phi(\mathcal{S}| \mathcal{H}^{'})
\end{equation}
\end{definition}
It is reasonable to consider that even if there is no tag, some default tag $h^{'}$ still exists, which can be used even if no tag is selected. We want to select $\ell$ many tags on top of the default tag. In Equation No. \ref{Eq:4}, $\mathcal{S}^{OPT}$ and $ \mathcal{H}^{OPT}$ denote the optimal slot subset of $k$ and the optimal tag subset of $\ell$. It can be easily observed that the problem introduced in Definition \ref{Def:Problem1} is the generalization of the \textsc{Influential Billboard Slot Selection Problem} \cite{ali2022influential,ali2023influential} where the context-dependent influence probability is not considered. Hence, Theorem \ref{Th:1} holds.

\begin{theorem} \label{Th:1}
For a given $k$ and $\ell$, finding the optimal slot and tag set for the Context-Dependent Influential Billboard Slot Selection Problem is \textsf{NP-hard}.
\end{theorem}

\section{Proposed Solution Approach} \label{Sec:Proposed}

\paragraph{\textbf{Exhaustive Search Approach.}}
In this approach, we enumerate all $k$-sized subsets of the set of billboard slots and $\ell$-sized subsets of tags. Considering all the billboards are running for the duration $[T_1, T_2]$ and the slot duration of $\Delta$ time units, hence the number of billboard slots will be $\frac{T_2-T_1+1}{\Delta} \cdot m$. So, the number $k$-sized subsets will be $\binom{\frac{T_2-T_1+1}{\Delta} \cdot m}{k}$ and the number of $\ell$ sized subsets of $\mathcal{H}$ will be $\binom{|\mathcal{H}|}{\ell}$. Subsequently, we create all possible $k$-sized slot subsets and $\ell$-sized tag subsets pairs, and for every possible slot-tag pair, we compute the influence and choose the one that gives the maximum influence and return it. 

\paragraph{\textbf{Orthant-Wise Greedy Maximization Algorithm.}}
In this approach, we start with a default slot $s^{'}$ and default tag $h^{'}$, and our approach is as follows. First, we fix the tag set to $\{h^{'}\}$ and apply an incremental greedy algorithm \cite{minoux2005accelerated,mirzasoleiman2015lazier} that works based on marginal gain computation to obtain the $k$ size slot set $\mathcal{S}^{'}$. Now, fixing the slot set to $\mathcal{S}^{'} \cup \{s^{'}\}$ , we apply incremental greedy algorithm to obtain the $\ell$ size tag set $\mathcal{H}^{'}$. Next, we do the same thing; however first fix the slot set to the default slot and apply the incremental greedy algorithm to choose an $\ell$ size tag set $\mathcal{H}^{''}$, and then we fix the tag set to $\mathcal{H}^{''} \cup \{h^{'}\}$ and apply the incremental greedy algorithm to obtain the $k$ size slot set $\mathcal{S^{''}}$. So we have two slot-tag pair $(\mathcal{S}^{'}, \mathcal{H}^{'})$ and $(\mathcal{S}^{''}, \mathcal{H}^{''})$. We return one that leads to the maximum influence. This method consists of the following four optimization problems.

\begin{equation}
\mathcal{S}^{'} \longleftarrow \underset{\mathcal{S} \subseteq \mathcal{BS} \wedge |\mathcal{S}|=k}{argmax} \ \phi(\mathcal{S} \cup \{s^{'}\}, \{h^{'}\})
\end{equation} 
\begin{equation}
\mathcal{H}^{'} \longleftarrow \underset{H \subseteq \mathcal{H} \wedge |H|=\ell}{argmax} \ \phi(\mathcal{S}^{'}, H \cup \{h^{'}\})
\end{equation} 

\begin{equation}
\mathcal{H}^{''} \longleftarrow \underset{H \subseteq \mathcal{H} \wedge |H|=\ell}{argmax} \ \phi(\{s^{'}\}, H \cup \{h^{'}\})
\end{equation}

\begin{equation}
\mathcal{S}^{''} \longleftarrow \underset{\mathcal{S} \subseteq \mathcal{BS} \wedge |\mathcal{S}|=k}{argmax} \ \phi(\mathcal{S} \cup \{s^{'}\},\mathcal{H}^{''})
\end{equation}

\paragraph{\textbf{Lazy Greedy Algorithm.}}
This approach involves excessive influence function evaluations, leading to high execution time. However, it can be implemented efficiently, with fewer evaluations in most practical cases, though the worst-case scenario matches the incremental greedy algorithm. The key idea is to consider the first for loop and its first iteration. We compute the marginal gain for all slots with respect to the empty tag set, which is equivalent to computing their influence value. Subsequently, we sort the slots based on this value in descending order, and the first slot is chosen. Now, in the second iteration, we compute the marginal gain of the slots in sorted order and consider the following situation. Suppose the marginal gain of the $i$-th slot is less than that of the $(i+1)$-th slot. Now, applying the submodularity property, it can be ensured that even if we compute the marginal gain of the slots, it can not be more than the marginal gain of the $i$-th slot. Hence, from the $(i+1)$-th slot onward, there is no need to compute their marginal gains, and they can be skipped safely. This improves the execution time, though the worst-case time complexity will remain the same. 

\paragraph{\textbf{Stochastic Greedy Algorithm.}}
In this approach \cite{mirzasoleiman2015lazier} in each iteration instead of computing the marginal gains of all the remaining elements, we sample $\frac{n}{k} \log \frac{1}{\epsilon}$ many elements from the ground set for slot selection and $\frac{n}{\ell} \log \frac{1}{\epsilon}$ many elements for tag selection. The marginal gain is computed only for the sampled elements. Here, we mention that $\epsilon$ is a control parameter that controls the trade-off between the quality of the solution and the execution time. Algorithm \ref{Algo:4} describes this process as pseudo-code.
\begin{algorithm}[h!]
\scriptsize
 \KwData{The Trajectory Database $\mathcal{D}$, The Billboard Database $\mathcal{B}$, Context Specific Influence Probabilities, Two Positive Integers $k$ and $\ell$.}
 \KwResult{$\mathcal{S} \subseteq V(G)$ with $|\mathcal{S}|=k$ and $\mathcal{H}^{'} \subseteq \mathcal{H}$ with $|\mathcal{H}^{'}|=\ell$ such that $\phi(\mathcal{S}, \mathcal{H}^{'})$ is maximized.}
 $\mathcal{S}^{'} \longleftarrow \{s^{'}\}$, $\mathcal{S}^{''} \longleftarrow \{s^{'}\}$, $\mathcal{H}^{'} \longleftarrow \{h^{'}\}$, $\mathcal{H}^{''} \longleftarrow \{h^{'}\}$\;
 $\mathcal{R} \longleftarrow \emptyset$, $\epsilon \longleftarrow 0.01$\;
 \For{$i=1 \text{ to }k$}{
 $\mathcal{R} \longleftarrow \text{Sample }\frac{a}{k} \log \frac{1}{\epsilon}\text{ many elements from }\mathcal{BS} \setminus \mathcal{S}^{'}$\;
 $s^{*} \longleftarrow \underset{s \in \mathcal{R}}{argmax} \ \phi(\mathcal{S}^{'} \cup \{s\},\{h^{'}\}) - \phi(\mathcal{S}^{'},\mathcal{H}^{'})$\;
 $\mathcal{S}^{'} \longleftarrow \mathcal{S}^{'} \cup \{s^{*}\}$\;
 }
 \For{$i=1 \text{ to }\ell$}{
 $\mathcal{R} \longleftarrow \text{Sample }\frac{b}{\ell} \log \frac{1}{\epsilon}\text{ many elements from }\mathcal{H} \setminus \mathcal{H}^{'}$\;
 $h^{*} \longleftarrow \underset{h \in \mathcal{R}}{argmax} \ \phi(\mathcal{S}^{'},\mathcal{H}^{'} \cup \{h\}) - \phi(\mathcal{S}^{'},\mathcal{H}^{'})$\;
 $\mathcal{H}^{''} \longleftarrow \mathcal{H}^{''} \cup \{h^{*}\}$\;
 }
 
 \For{$i=1 \text{ to }\ell$}{
 $\mathcal{R} \longleftarrow \text{Sample }\frac{b}{\ell} \log \frac{1}{\epsilon}\text{ many elements from }\mathcal{H} \setminus \mathcal{H}^{''}$\;
 $h^{*} \longleftarrow \underset{h \in \mathcal{R}}{argmax} \ \phi(\mathcal{S}^{''},\mathcal{H}^{''} \cup \{h\}) - \phi(\mathcal{S}^{''},\mathcal{H}^{''})$\;
 $\mathcal{H}^{''} \longleftarrow \mathcal{H}^{''} \cup \{h^{*}\}$\;
 } 
 
\For{$i=1 \text{ to }k$}{
$\mathcal{R} \longleftarrow \text{Sample }\frac{a}{k} \log \frac{1}{\epsilon}\text{ many elements from }\mathcal{BS} \setminus \mathcal{S}^{''}$\;
 $s^{*} \longleftarrow \underset{s \in \mathcal{R}}{argmax} \ \phi(\mathcal{S}^{''} \cup \{s\},\mathcal{H}^{''}) - \phi(\mathcal{S}^{''},\mathcal{H}^{''})$\;
 $\mathcal{S}^{''} \longleftarrow \mathcal{S}^{''} \cup \{s^{*}\}$\;
 } 
 \eIf{$\phi(\mathcal{S}^{'},\mathcal{H}^{'}) > \phi(\mathcal{S}^{''},\mathcal{H}^{''})$}{
$return \ \mathcal{S}^{'} \text{ and } \mathcal{H}^{'}$\;
 }
 {
 $return \ \mathcal{S}^{''} \text{ and } \mathcal{H}^{''}$
 }
 \caption{Stochastic Greedy Algorithm for the Influential Slots and Tags Selection Problem}
 \label{Algo:4}
\end{algorithm}

\par \textbf{Complexity Analysis.} Now, we analyze the time and space requirements for Algorithm \ref{Algo:4}. Initialization at Line No. $1$ and $2$ will take $\mathcal{O}(1)$ time. To sample out $\frac{a}{k} \log \frac{1}{\epsilon}$ many element it will take $\mathcal{O}(a. \log \frac{1}{\epsilon})$ time. Now, for any billboard slot $s \in \mathcal{BS}$ and $l \in \mathcal{H}$, calculating influence using equation \ref{Eq:Equation_Inf} will take $\mathcal{O}(t)$ time, in which $t$ is the number of tuple in the trajectory database. In Line No. $5$ computing marginal gain will take $\mathcal{O}(2.a. \log \frac{1}{\epsilon}.t)$ time and Line No. $6$ will execute for $\mathcal{O}(k)$ time. So, Line No. $3$ to $6$ will take $\mathcal{O}(a. \log \frac{1}{\epsilon} + 2.a. \log \frac{1}{\epsilon}.t + k)$ time. In Line No. $7$ to $10$ will take $\mathcal{O}(b. \log \frac{1}{\epsilon} + 2.b. \log \frac{1}{\epsilon}.k.t + \ell)$ time and Line No. $11$ to $14$ will take $\mathcal{O}(b. \log \frac{1}{\epsilon} + 2.b. \log \frac{1}{\epsilon}.\ell.t + \ell)$. In the fourth greedy time taken by Line No. $15$ to $18$ is of $\mathcal{O}(a. \log \frac{1}{\epsilon} + 2.a. \log \frac{1}{\epsilon}.\ell.t + k)$. Finally, Line No. $19$ to $22$ will take $\mathcal{O}(2.k.\ell.t)$ time for final comparison. Hence, total time requirement of Algorithm \ref{Algo:4} will be $\mathcal{O}(a. \log \frac{1}{\epsilon}.\ell.t + b. \log \frac{1}{\epsilon}.k.t + k.\ell.t)$. Now, the additional space requirement to store the lists $\mathcal{S}^{'}, \mathcal{S}^{''}, \mathcal{H}^{'}$, $ \mathcal{H}^{''}$ and $\mathcal{R}$ will be $\mathcal{O}(k),\mathcal{O}(k), \mathcal{O}(\ell)$, $ \mathcal{O}(\ell)$ and $\mathcal{O}(max(a. \log \frac{1}{\epsilon},b.\log \frac{1}{\epsilon}))$ respectively. Hence, total space requirement for Algorithm \ref{Algo:4} will be of $\mathcal{O}(max(a. \log \frac{1}{\epsilon},b.\log \frac{1}{\epsilon}) + 2k+2\ell)$. 

\par Now, we analyze this methodology and prove some theoretical results. 
\begin{lemma} \label{lemma:4}
The number of influence function evaluations by Algorithm \ref{Algo:4} will be equal to $4(a+b) \log \frac{1}{\epsilon}$, i.e.,  $\mathcal{O}((a+b) \log \frac{1}{\epsilon})$.
\end{lemma}

\begin{lemma} \label{Lemma:5}
Consider the first \texttt{for loop} of Algorithm \ref{Algo:4} and assume that after the execution of its $i$-th iteration, the solution set is $\mathcal{S}^{'}_{i}$. The expected influence gain of Algorithm \ref{Algo:4} in the $(i+1)$-th will be at least $\frac{1-\epsilon}{k} \underset{s^{*} \in \mathcal{S}^{OPT} \setminus \mathcal{S}^{'}}{\sum} \phi(\mathcal{S}^{'} \cup \{s^{*}\},\{\mathcal{H}^{'}\})- \phi(\mathcal{S}^{'},\mathcal{H}^{'})$.
\end{lemma}

\begin{theorem} \label{Th:3}
Let $\mathcal{S}^{OPT}$ and $\mathcal{H}^{OPT}$ be an optimal $k$-sized and an $\ell$-sized slot and tag set, respectively. Also assume $\mathcal{S}^{\mathcal{A}}$ and $\mathcal{H}^{\mathcal{A}}$ are the $k$-sized and an $\ell$-sized slot and tag set returned by Algorithm \ref{Algo:4}. Then $\phi(\mathcal{S}^{\mathcal{A}}, \mathcal{H}^{\mathcal{A}}) \geq (1-\frac{1}{e} - \epsilon)^{2} \cdot \phi(\mathcal{S}^{OPT}, \mathcal{H}^{OPT})$. In other words, Algorithm \ref{Algo:4} gives $(1-\frac{1}{e} - \epsilon)^{2}$ factor approximation guarantee.
\end{theorem}

\begin{proof}
It can be observed that any one of the following two cases may happen. \textbf{Case I:} $\mathcal{S}^{\mathcal{A}}=\mathcal{S}^{'}$, $\mathcal{H}^{\mathcal{A}}=\mathcal{H}^{'}$ and \textbf{Case II:} $\mathcal{S}^{\mathcal{A}}=\mathcal{S}^{''}$, $\mathcal{H}^{\mathcal{A}}=\mathcal{H}^{''}$

Let, $\mathcal{S}^{'}_{i} = \{s_{1}, s_{2}, s_{3},\ldots s_{i}\}$ defines the solutions at each step  returns by first \texttt{For Loop} in Algorithm \ref{Algo:4} after $i^{th}$ iteration. Now, from lemma \ref{Lemma:5} we can write, 
\begin{align}
E[\Delta(s_{i+1}|\mathcal{S}^{'},\mathcal{H}^{'})|\mathcal{S}^{'},\mathcal{H}^{'}] \geq \frac{1-\epsilon}{k} \underset{s^{*} \in \mathcal{S}^{OPT} \setminus \mathcal{S}^{'}}{\sum}\Delta({s}^{*} | \mathcal{S}^{'},\mathcal{H}^{'})
\end{align}
Using the submodularity property, we can obtain,
\begin{align}
\underset{s^{*} \in \mathcal{S}^{OPT} \setminus \mathcal{S}^{'}}{\sum}\Delta({s}^{*} | \mathcal{S}^{'},\mathcal{H}^{'}) & \geq \Delta(\mathcal{S}^{OPT} | \mathcal{S}^{'}_{i}, \mathcal{H}^{'})\nonumber \\ & \geq \phi(\mathcal{S}^{OPT}, \mathcal{H}^{'})- \phi(\mathcal{S}^{'}_{i}, \mathcal{H}^{'}) \nonumber
\end{align}
Now, if we put these results in Equation (14), we get,

\begin{align}
E[\phi(\mathcal{S}^{'}_{i+1},\mathcal{H}^{'})-\phi(\mathcal{S}^{'}_{i},\mathcal{H}^{'})|\mathcal{S}^{'}_{i}] & \geq \frac{1-\epsilon}{k}\phi(\mathcal{S}^{OPT}, \mathcal{H}^{'})- \phi(\mathcal{S}^{'}_{i}, \mathcal{H}^{'}) \nonumber
\end{align}

Now, if we take expectation over $\mathcal{S}^{'}_{i}$, we can obtain,
\begin{align}
E[\phi(\mathcal{S}^{'}_{i+1},\mathcal{H}^{'})-\phi(\mathcal{S}^{'}_{i},\mathcal{H}^{'})] = \frac{1-\epsilon}{k}\phi(\mathcal{S}^{OPT}, \mathcal{H}^{'})- \phi(\mathcal{S}^{'}_{i}, \mathcal{H}^{'}) \nonumber]
\end{align}
If we apply induction to it,
\begin{align}
E[\phi(\mathcal{S}^{'}_{k}, \mathcal{H}^{'})] \geq (1-(1-\frac{1-\epsilon}{k})^{k}).\phi (\mathcal{S}^{OPT},\mathcal{H}^{\mathcal{A}})\nonumber & \\ \geq (1-\frac{1}{e}-\epsilon). \phi (\mathcal{S}^{OPT},\mathcal{H}^{\mathcal{A}})
\end{align}
Now, in a similar way, for the second, third, and fourth \texttt{For Loop}, we can write:
\begin{align}
E[\phi(\mathcal{S}^{'}_{k}, \mathcal{H}^{'}_{\ell})] \geq (1-\frac{1}{e}-\epsilon)^{2}. \phi (\mathcal{S}^{OPT},\mathcal{H}^{OPT})
\end{align}

\begin{align}
E[\phi(\mathcal{S}^{''}, \mathcal{H}^{''}_{\ell})] \geq (1-\frac{1}{e}-\epsilon). \phi (\mathcal{S}^{\mathcal{A}},\mathcal{H}^{OPT})
\end{align}

\begin{align}
E[\phi(\mathcal{S}^{''}_{k}, \mathcal{H}^{''}_{\ell})] \geq (1-\frac{1}{e}-\epsilon)^{2}. \phi (\mathcal{S}^{OPT},\mathcal{H}^{OPT})
\end{align}
\end{proof}

\section{Experimental Evaluations} \label{Sec:EE}
This section describes the experimental evaluations of the proposed solution approaches. Initially, we start by describing the datasets used in our experiments.
\paragraph{\textbf{Dataset Description.}}
We use two widely studied datasets for our experiments \cite{ali2023influential,zhang2020towards}. The first dataset includes 227,428 check-in records from New York City\footnote{\url{https://www.nyc.gov/site/tlc/about/tlc-trip-record-data.page}}, collected over ten months (April 12, 2012–February 16, 2013), with details like timestamps, GPS coordinates, and user IDs. The second dataset, VehDS-LA\footnote{\url{https://github.com/Ibtihal-Alablani}}, contains 74,170 vehicle records from 15 streets in Los Angeles, featuring street names, GPS coordinates, and timestamps. Additionally, billboard data from LAMAR\footnote{\url{http://www.lamar.com/InventoryBrowser}} includes billboard ID, venue ID, GPS coordinates, timestamps, and panel size. The New York City dataset has 716 billboards (1,031,040 slots), and Los Angeles has 1,483 billboards (2,135,520 slots).

\paragraph{\textbf{Key Parameters.}}
All the parameters are summarized in Table \ref{Table-2}, including the number of billboard slots $k$ and tags $\ell$ to be picked. The user-defined parameter $\epsilon$ defines the size of random subsets. The distance threshold, $\lambda$, determines the maximum distance a billboard can influence the trajectories. In each experiment, we fixed one parameter value and varied the other parameter values. All codes are executed in Python using Jupyter Notebook in an HP Z4 workstation with 64 GB of memory and an Xeon(R) 3.50 GHz processor.
\vspace{-0.3in}
\begin{table}[h!]
\begin{center}
\caption{\label{Table-2} Parameter Settings}
%\vspace{-0.1 in}
    \begin{tabular}{ | p{2cm}| p{5cm}|}
    \hline
    Parameter & Values  \\ \hline
    $k$ & $25, 50, 100, 150, 200$   \\ \hline
    $\ell$ & $10, 20, 30, 40, 50$  \\ \hline
    $\epsilon$ & $0.01, 0.05, 0.1, 0.15, 0.2$ \\ \hline
    $\lambda$ & $25m, 50m, 75m, 100m, 125m$  \\ \hline
    \end{tabular}
\end{center}
\end{table}
\vspace{-0.4in}

\paragraph{\textbf{Baseline Methodologies.}}
We compared our proposed solutions with the following baseline methods:

\par \textbf{Random Slot and Random Tag (RSRT):} In this method, $k$ many random slots and $\ell$ many random tags are chosen and returned as solution.
\par \textbf{Random Slot and High-Frequency Tag(RSHFT):} Tag frequency is defined by the number of associated people. We count and sort tags by frequency, then return $\ell$ tags and $k$ random slots from the sorted list.
\par \textbf{Maximum Coverage Slot and Random Tag (MAXSRT):} The coverage of a billboard slot is the number of people passing by it. We compute and sort the coverage for all slots, then return $k$ slots from the sorted list and select $\ell$ tags uniformly at random.
\par \textbf{Top-$k$ Slot and Top-$\ell$ Tag (TSTT):} This method calculates the individual influence of each billboard slot and tag, then sorts both in descending order. From the sorted lists, we select the Top-$k$ billboard slots and Top-$\ell$ tags.

\par \textbf{Top-$k$ Slot and Random $\ell$ Tag (TSRT):} In this method, the influence of each billboard slot is calculated, and the slots are sorted in descending order, with the Top-$k$ selected. For tags, $\ell$ is randomly chosen from the unsorted list.

\textbf{Random $k$ Slot and Top-$\ell$ Tag (RSTT):} This method is the reverse of the TSRT approach. First, the influence of each billboard slot and tag is calculated. Tags are then sorted in descending order by influence, and the Top-$\ell$ tags are selected. From the billboard slots, $k$ are randomly chosen.
\paragraph{\textbf{Goals of our Experiments.}} \label{exp_goal}
In this study, we address the following Research Questions (RQ).
\begin{itemize}
\item \textbf{RQ1}: How does the influence value increase if we increase the number of slots and tags to be selected?
\item \textbf{RQ2}: If we increase the number of slots and tags, how do the computational time requirements of the proposed and the baseline methods change?
\item \textbf{RQ3}: If we increase the size of the trajectory, how do the proposed method's influence value and computational time requirement change?  
\item \textbf{RQ4}: For the stochastic greedy algorithm, if we change the value of $\epsilon$, how do the computational time and the quality of the solution change?
\end{itemize}
\subsection{Experimental Results with Discussions}
In this section, we describe the experimental results and answer each research question posed in this work.
\paragraph{\textbf{Budget $(k,\ell)$ Vs. Influence.}} 
Budget and influence are critical factors in billboard advertisement decisions. In our experiment, we analyzed the influence of varying billboard slots $(25, 50, 100, 150,$ and $200)$ for different tag values $\ell$, shown in Figure \ref{Fig:Budget_Vs_INF_Time}. The influence probability of `tags' in the NYC dataset is unevenly distributed, with a few tags being highly influential while most are not. This distribution favors algorithms like `Lazy Greedy', `Stochastic Greedy', and baseline methods such as `TSTT', `RSTT', and `RSHFT'. In contrast, `MAXSRT' and `TSRT' underperform due to random tag selection. Conversely, the LA dataset exhibits a more balanced influence distribution, leading to better performance for `TSRT' and `RSTT'. In the LA dataset among the baseline methods, `TSTT' has almost equal influence to `Stochastic Greedy'. On the other hand, in the NYC dataset, the influence probability of billboard slots is well distributed, and the influence difference between `Stochastic Greedy' and `TSTT' is differentiable, as shown in Figure \ref{Fig:Budget_Vs_INF_Time} $(a,b,c,d,e)$. Now, when we increase the number of billboard slot from $25$ to $200$ with a fixed value of $\ell = 10$, $\epsilon =0.01$, the influence value of `Lazy Greedy', `Stochastic Greedy', `TSTT', `RSTT', `RSHFT', `MAXSRT', and `TSRT' are increases from $353.74$, $353.36$, $352.27$, $339.50$, $265.43$, $50.93$, $19.71$ to $437.55$, $434.21$, $414.16$, $376.49$, $297.14$, $84.25$, $56.19$ respectively. Similarly, if we fixed the number of billboard slot, $k = 200$ and vary $\ell$ from $10$ to $50$ then the influence value of `Lazy Greedy', `Stochastic Greedy', `TSTT', `RSTT', `RSHFT', `MAXSRT', and `TSRT' increase from $437.55$, $434.21$, $414.16$, $376.49$, $297.14$, $84.25$, $56.19$ to $651.45$, $641.69$, $617.26$, $583.42$, $577.81$, $338.21$, $294.46$ respectively. Similar types of observations were also observed in the LA dataset. Therefore, among the proposed two methods, `Lazy Greedy' gives more influence compared to `Stochastic Greedy' because of the randomized element selection behavior of `Stochastic Greedy', and `TSTT' gives maximum influence among other baseline methods for both LA and NYC datasets as reported in Figure \ref{Fig:Budget_Vs_INF_Time}.  

\begin{figure*}[h!]
\centering
\begin{tabular}{cccc}
\includegraphics[width=0.29\linewidth]{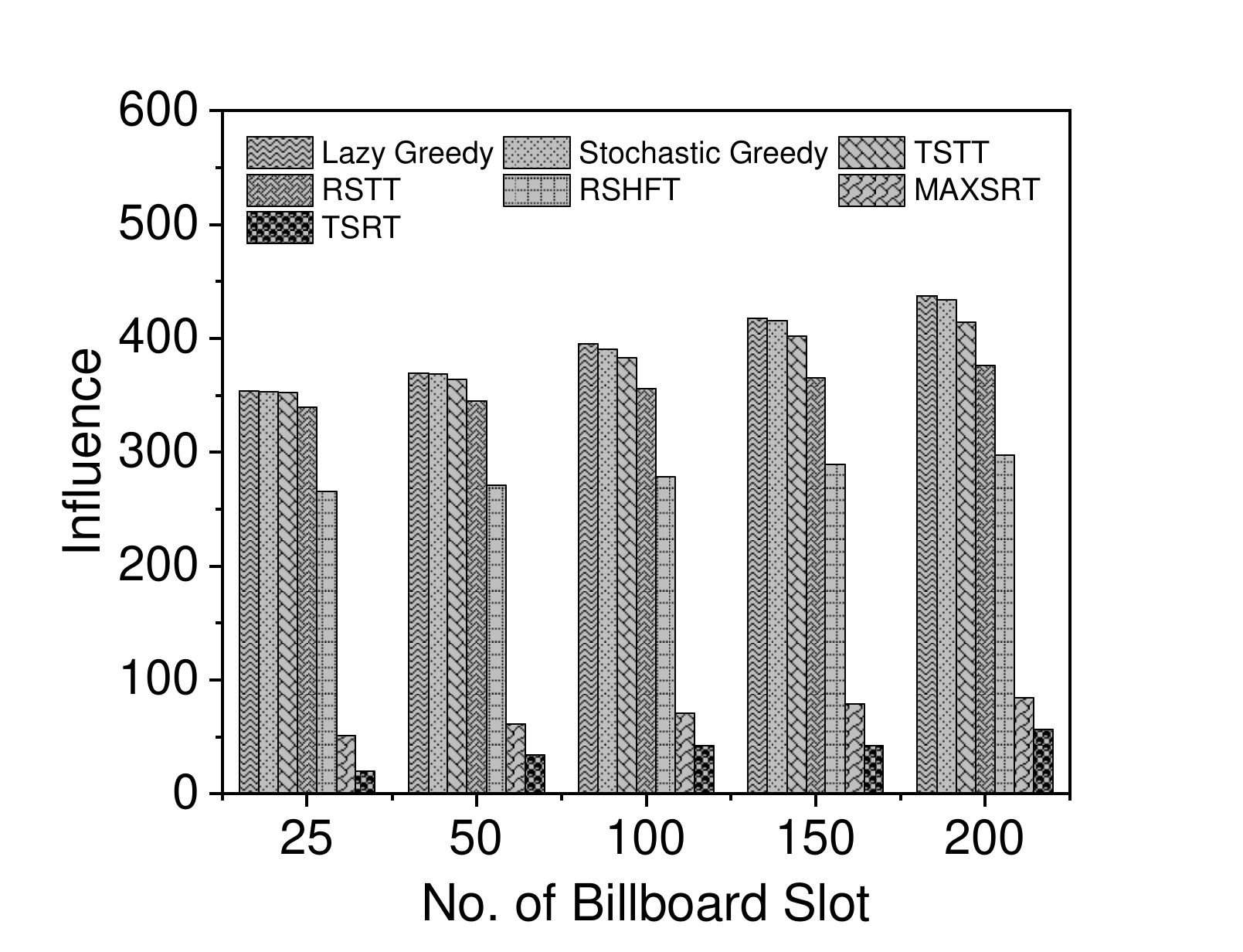}\hspace{-2em} & 
\includegraphics[width=0.29\linewidth]{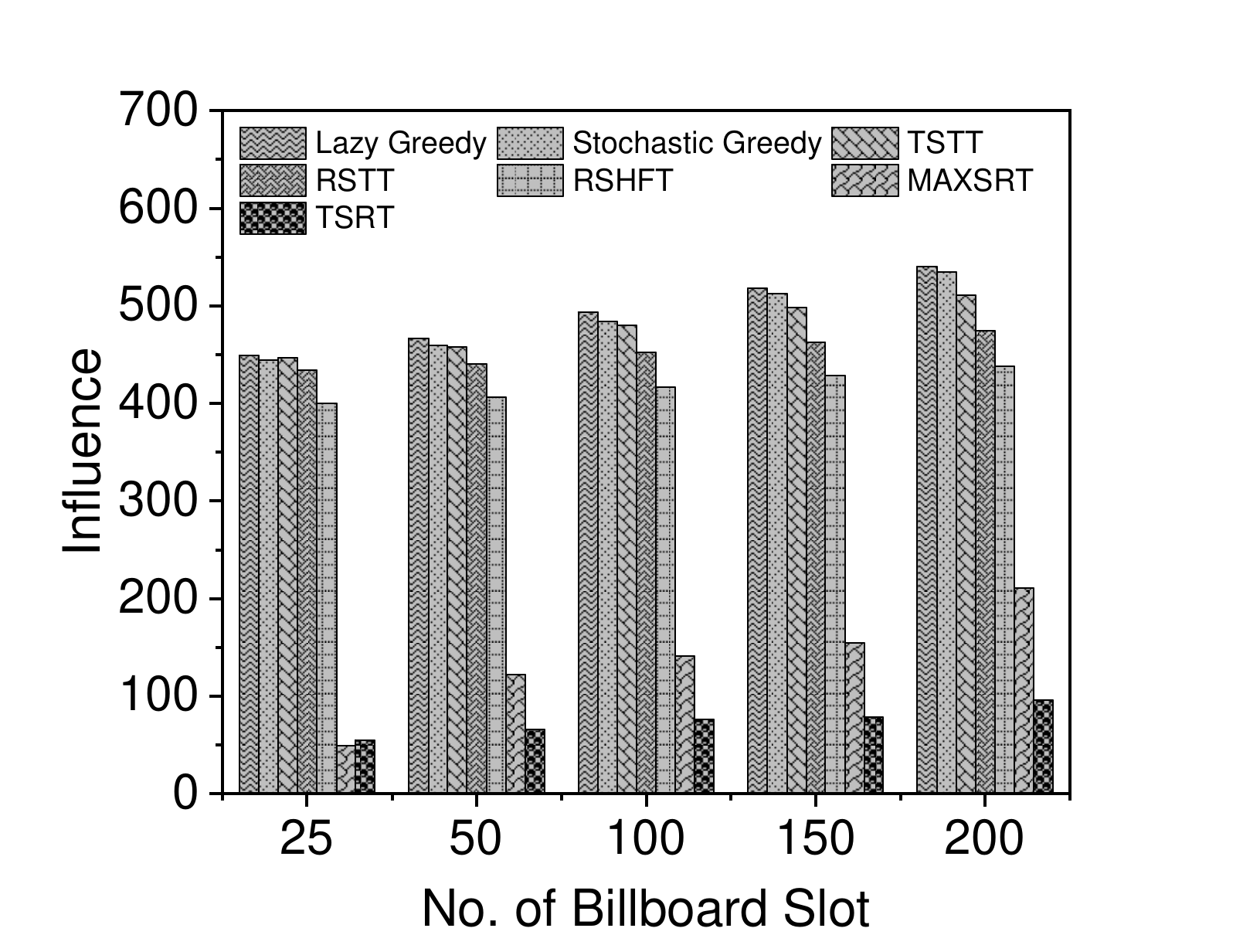}\hspace{-2em} & 
\includegraphics[width=0.29\linewidth]{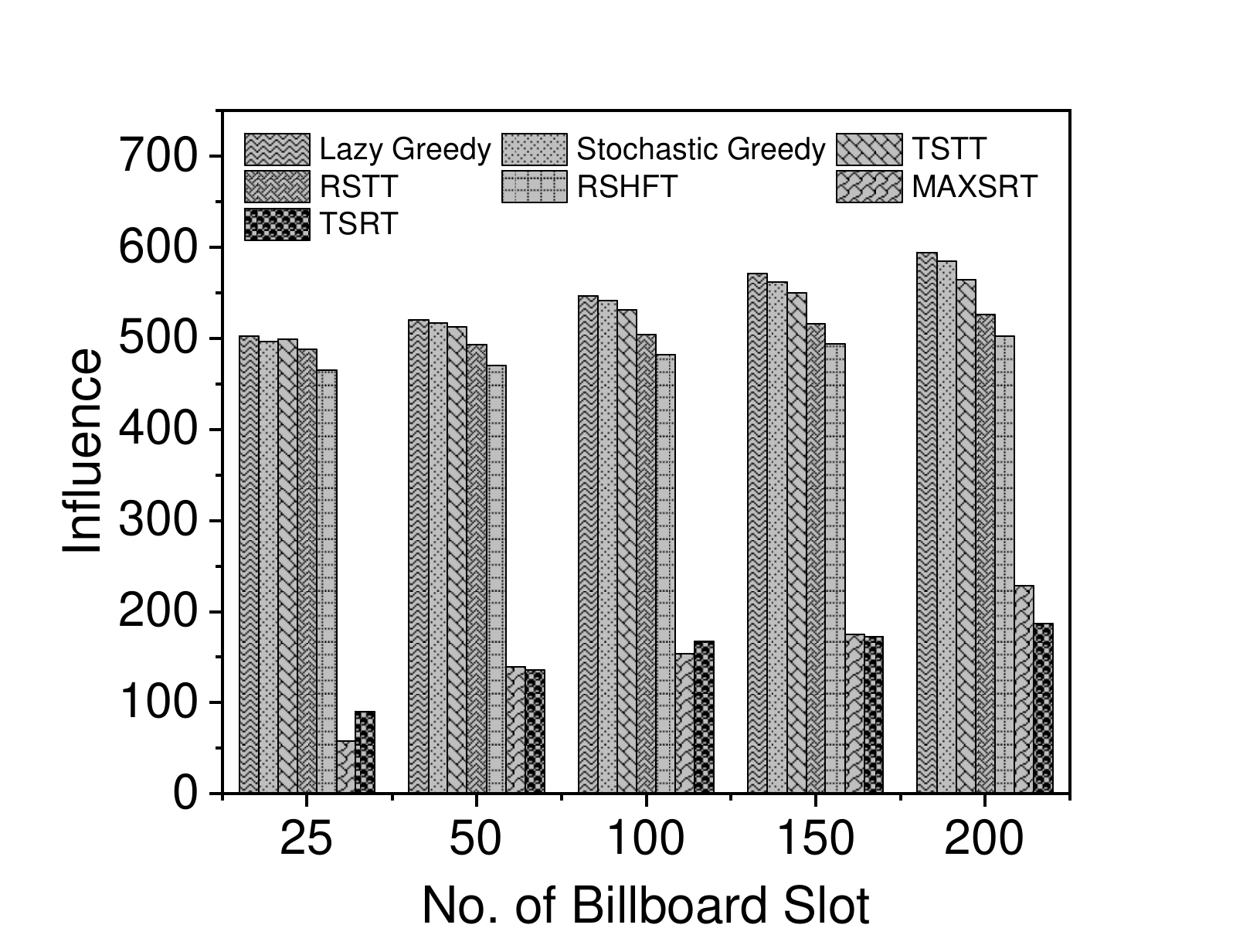}\hspace{-2em} &
\includegraphics[width=0.29\linewidth]{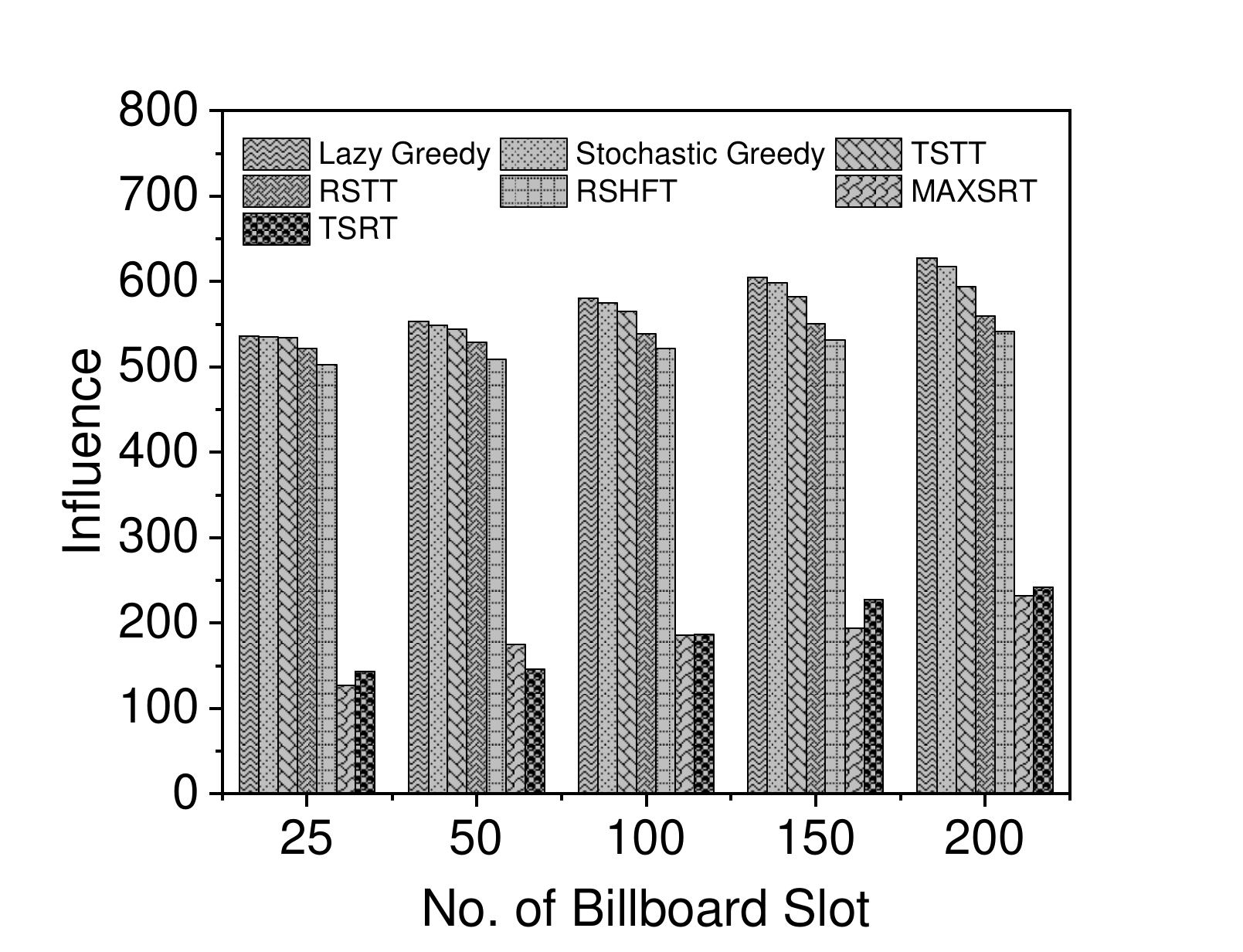} \\ 
{\tiny(a) $\ell = 10 $} & {\tiny(b) $\ell = 20 $} &  {\tiny(c) $\ell = 30$} &  {\tiny(d) $\ell = 40$} \\
\vspace{-0.5em}
\includegraphics[width=0.29\linewidth]{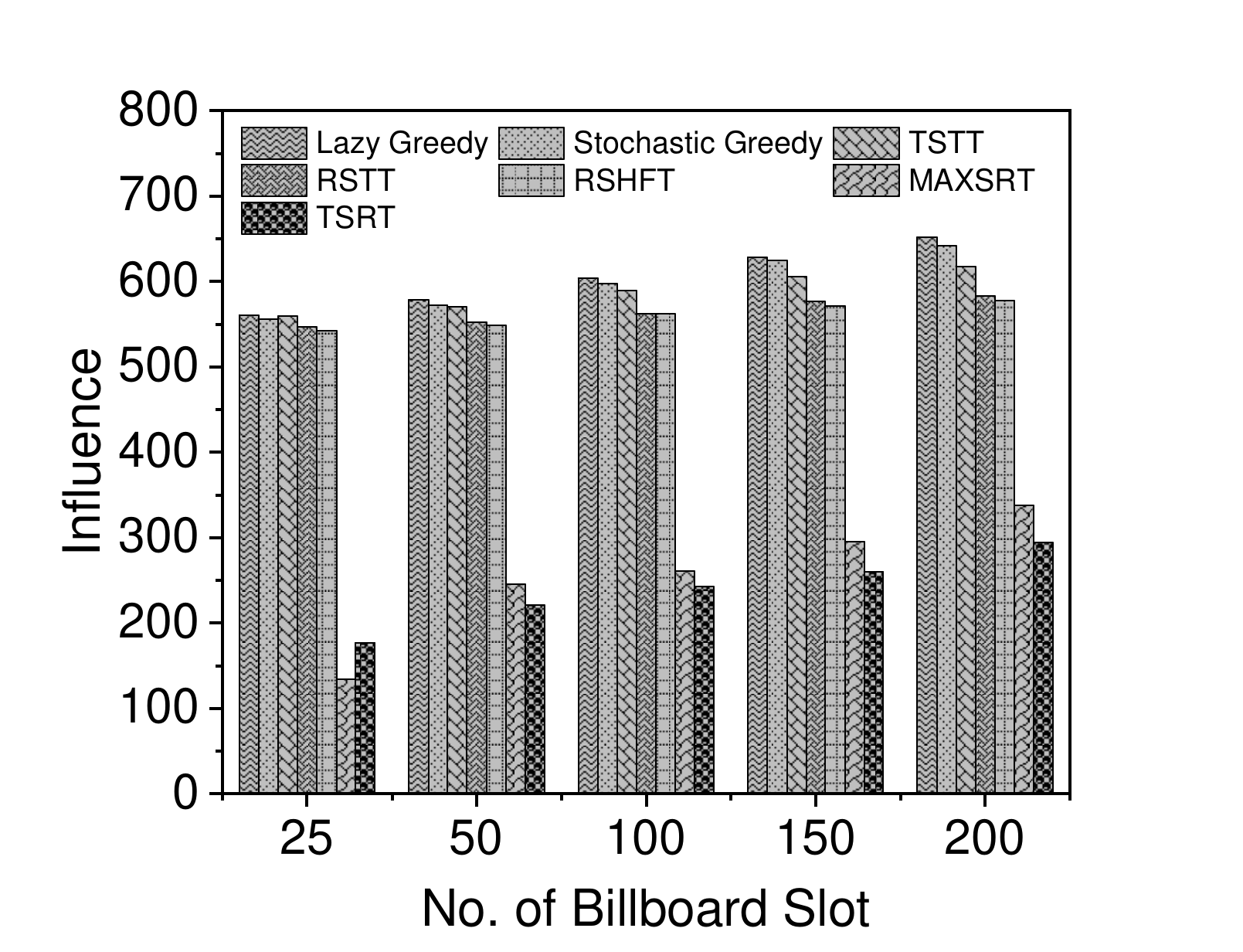}\hspace{-2em} &
\includegraphics[width=0.29\linewidth]{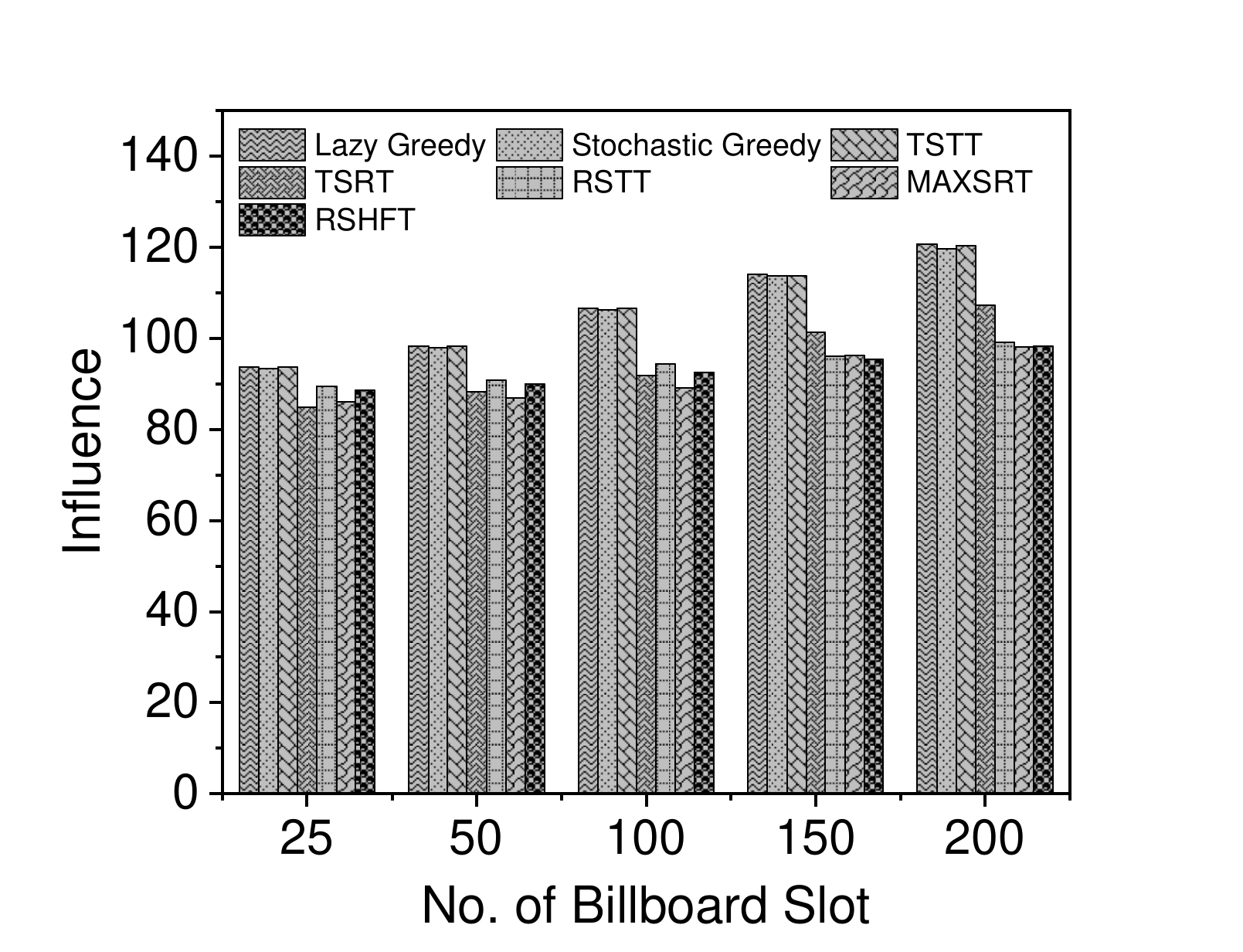}\hspace{-2em} & 
\includegraphics[width=0.29\linewidth]{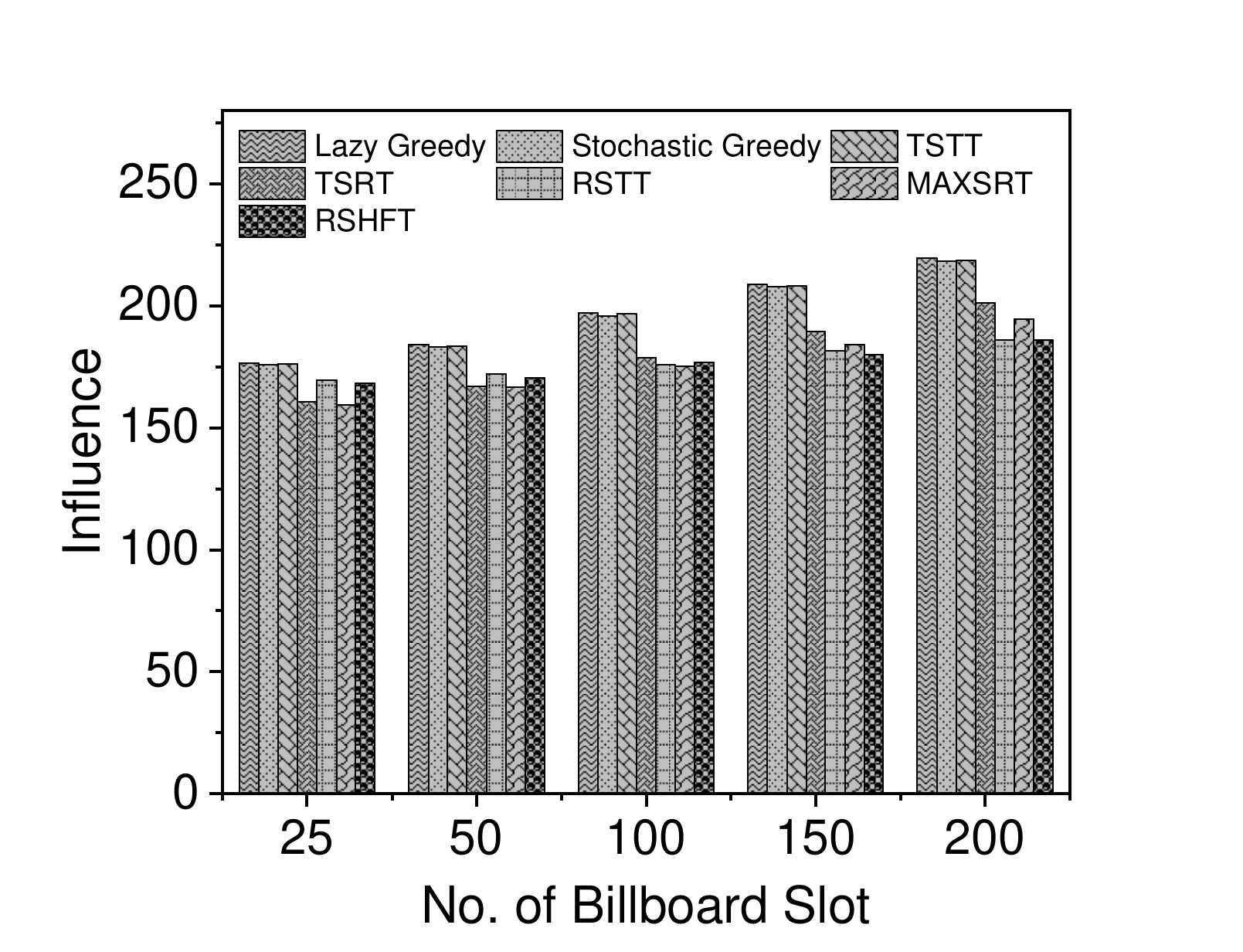}\hspace{-2em} & 
\includegraphics[width=0.29\linewidth]{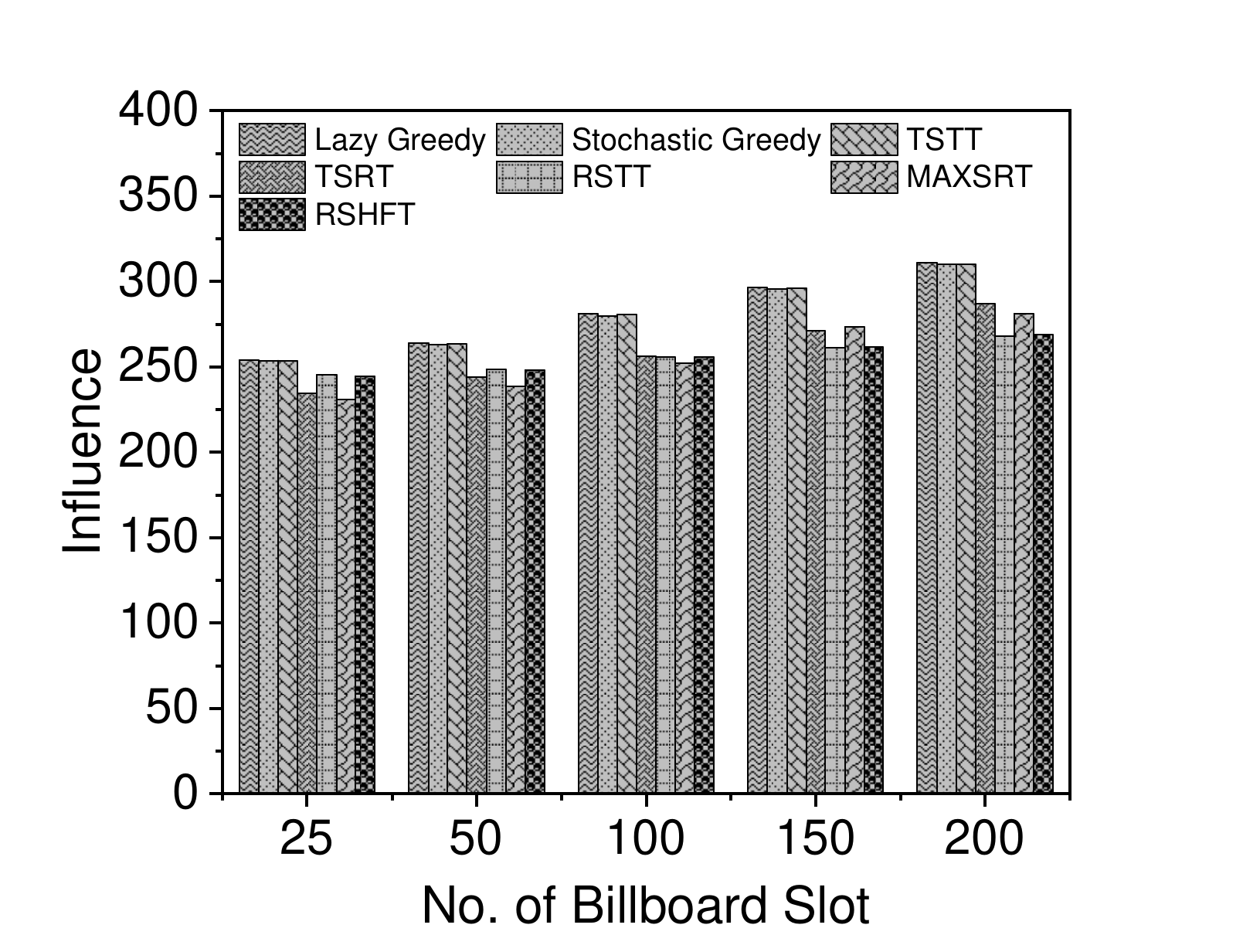} \\
{\tiny(e) $\ell = 50 $} &  {\tiny(f) $\ell = 10 $} &  {\tiny(g) $\ell = 20$} &  {\tiny(h) $\ell = 30$} \\
\vspace{-0.5em}
\includegraphics[width=0.29\linewidth]{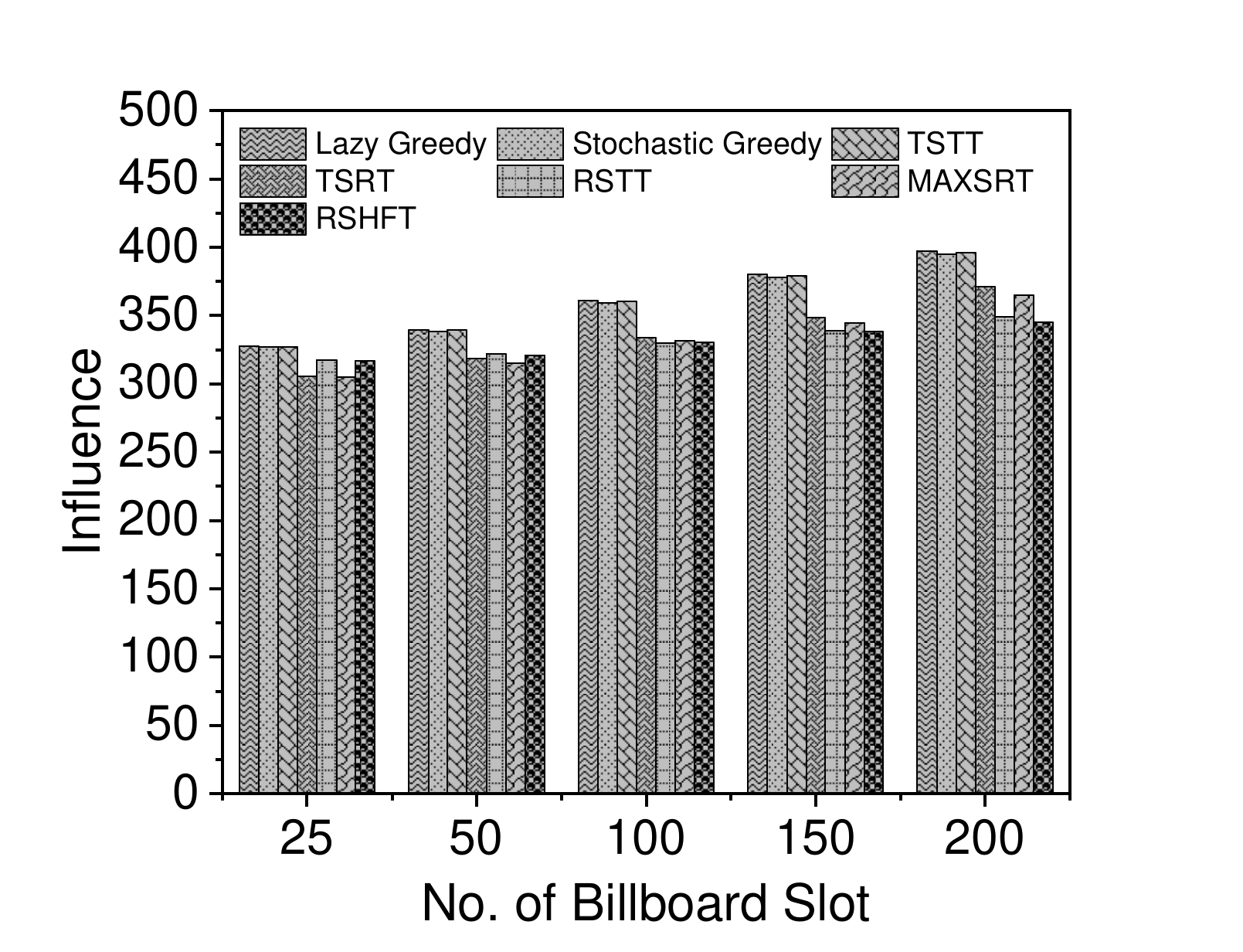}\hspace{-2em} & 
\includegraphics[width=0.29\linewidth]{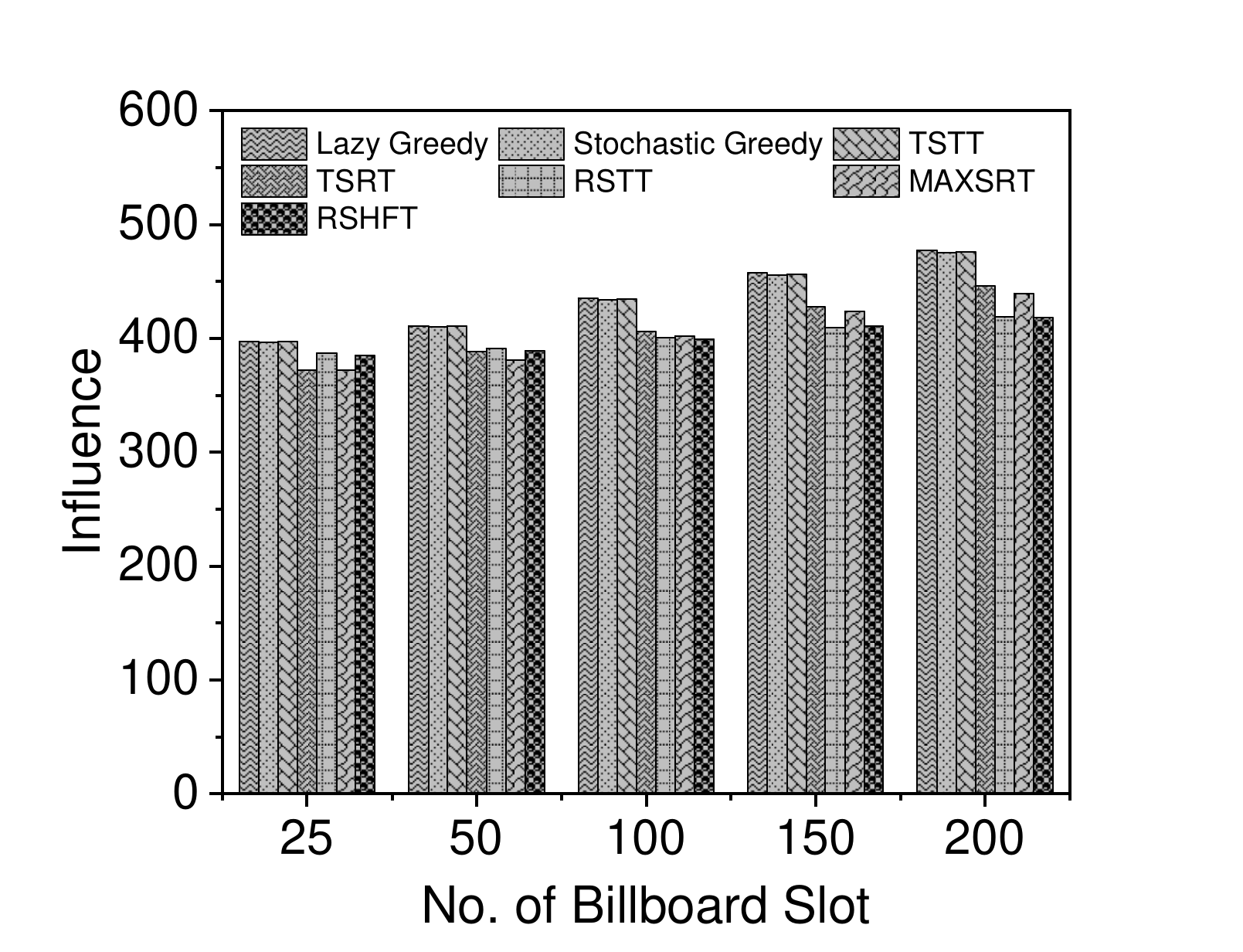}\hspace{-2em} &
\includegraphics[width=0.29\linewidth]{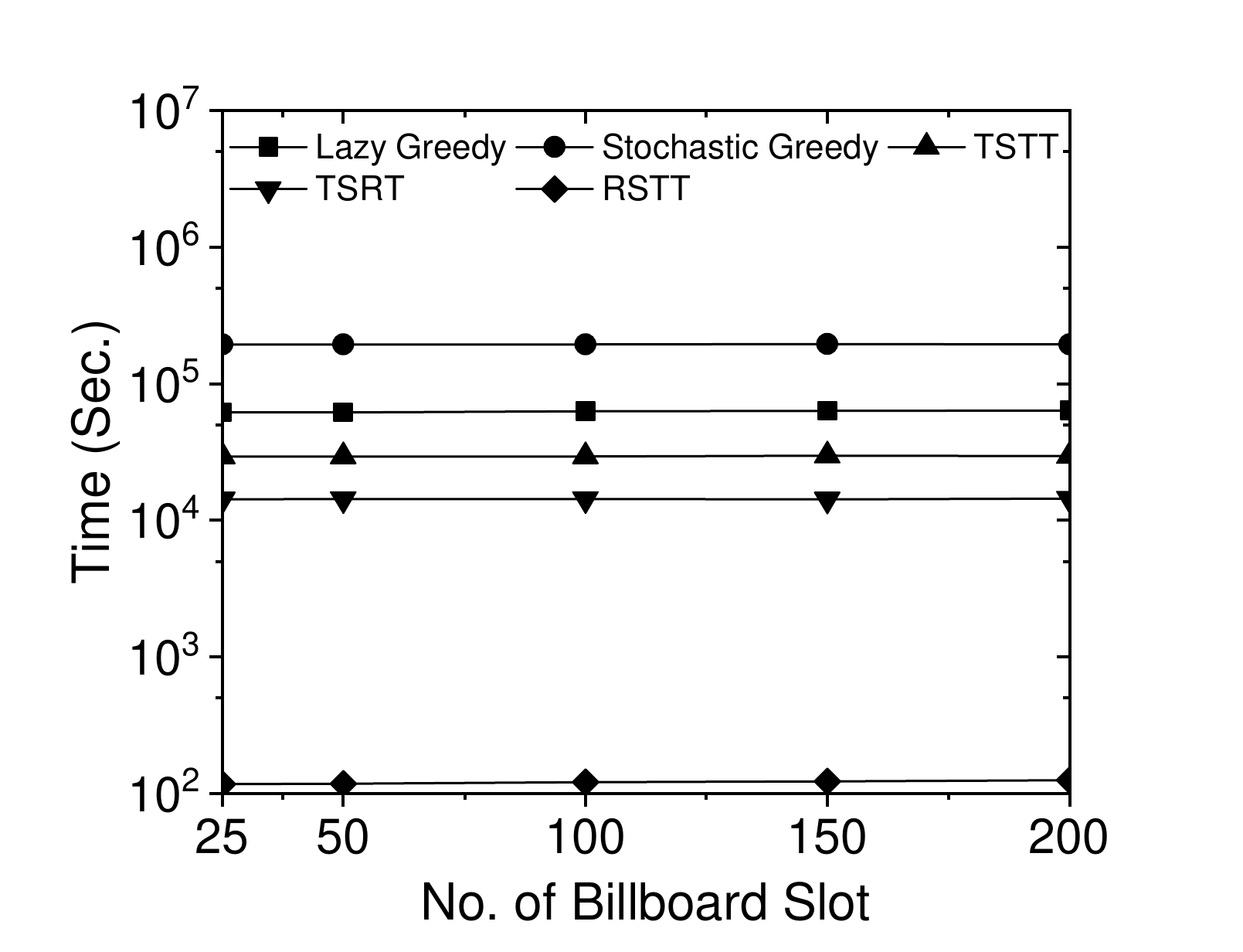}\hspace{-2em} & 
\includegraphics[width=0.29\linewidth]{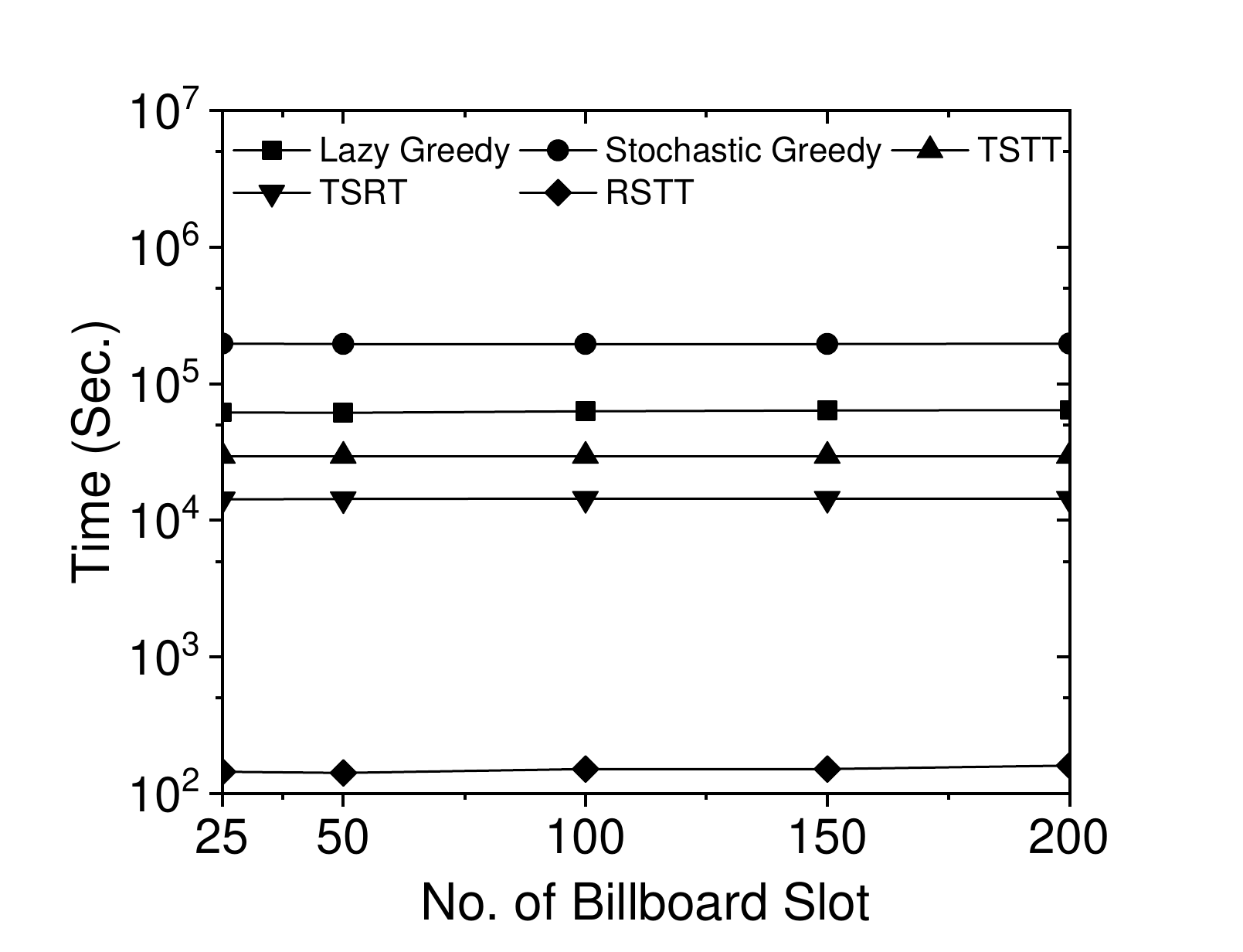} \\
 {\tiny(i) $\ell = 40 $} &  {\tiny(j) $\ell = 50 $} &  {\tiny(k) $\ell = 10$} &  {\tiny(l) $\ell = 20$} \\
\vspace{-0.5em}
\includegraphics[width=0.29\linewidth]{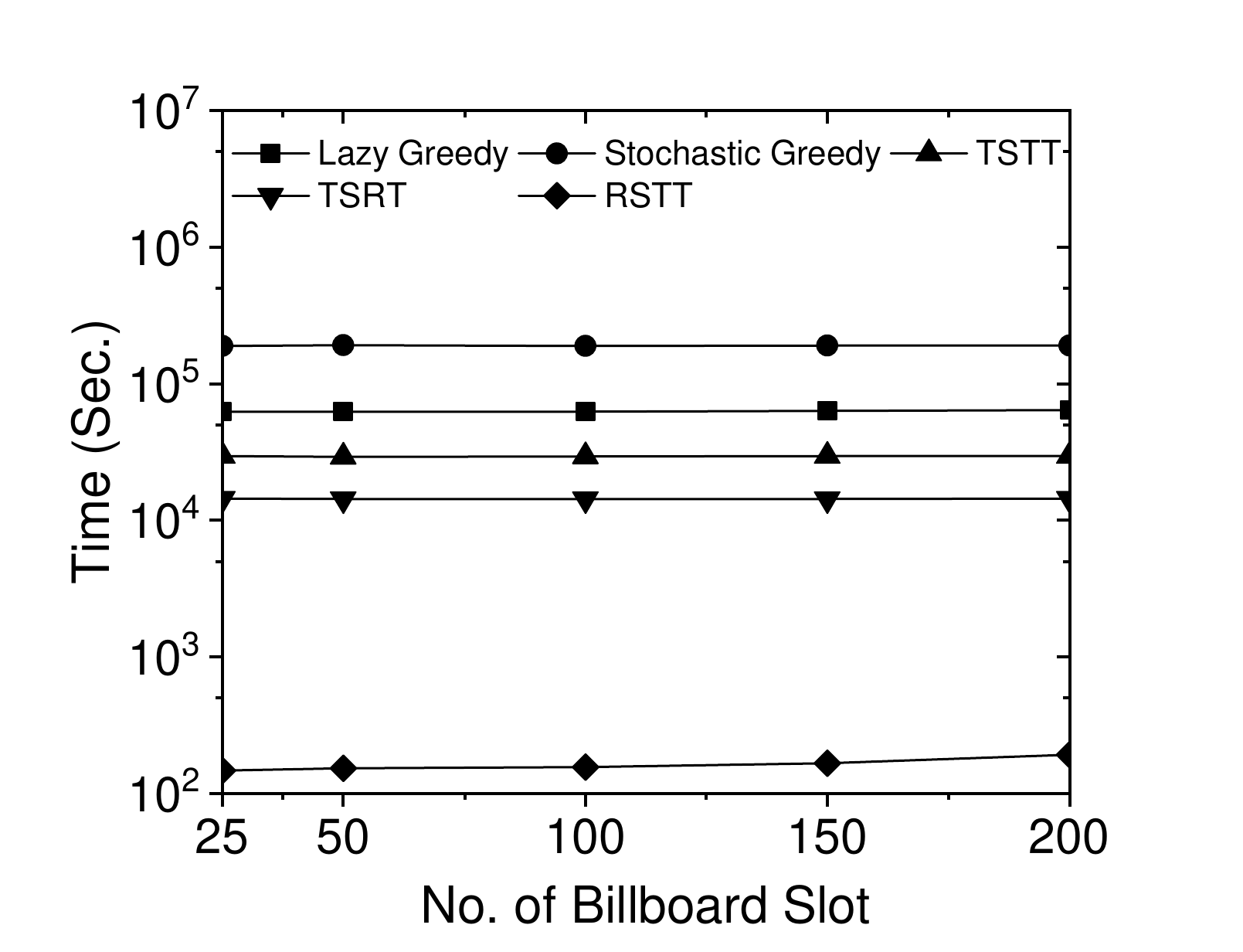}\hspace{-2em} & 
\includegraphics[width=0.29\linewidth]{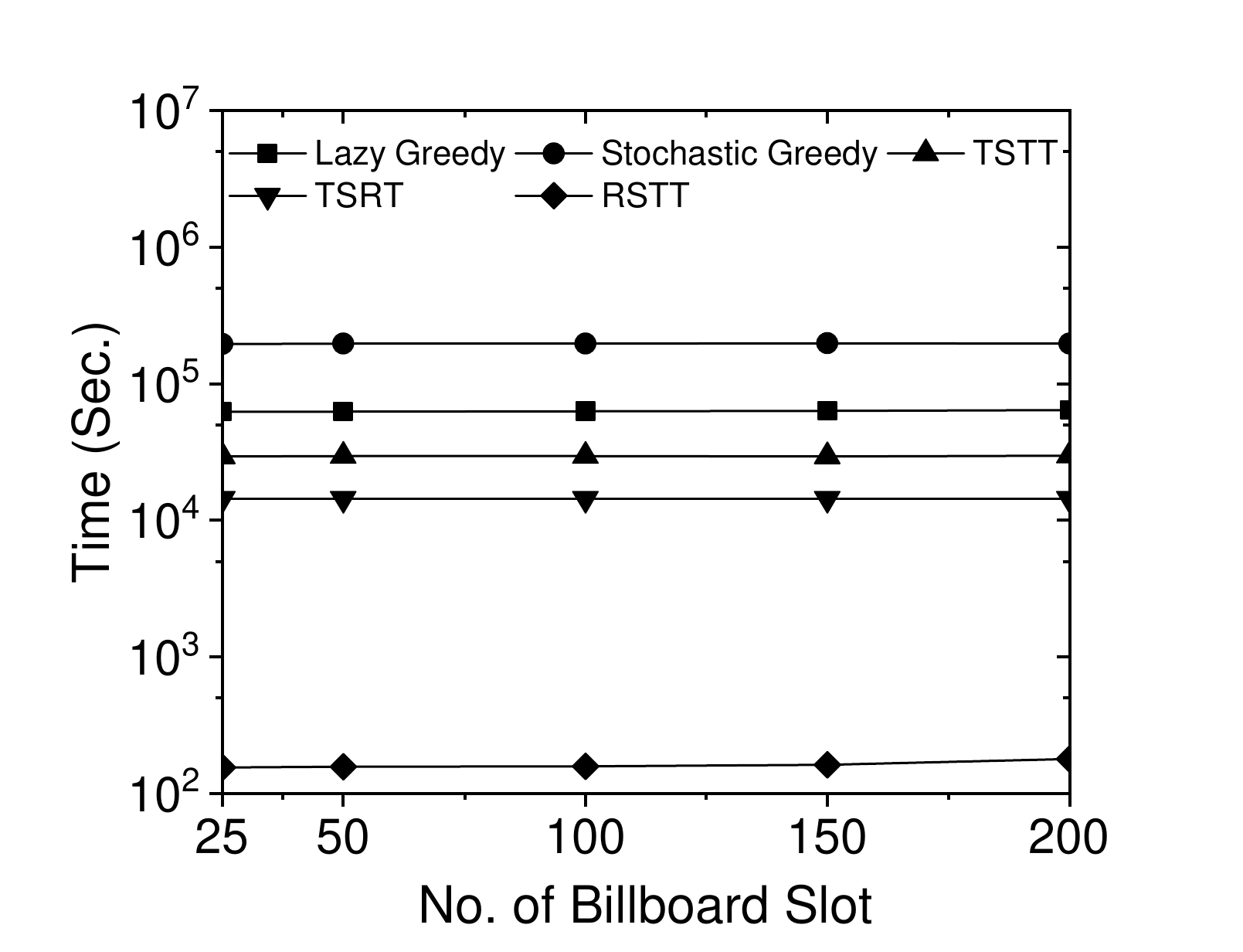}\hspace{-2em} &
\includegraphics[width=0.29\linewidth]{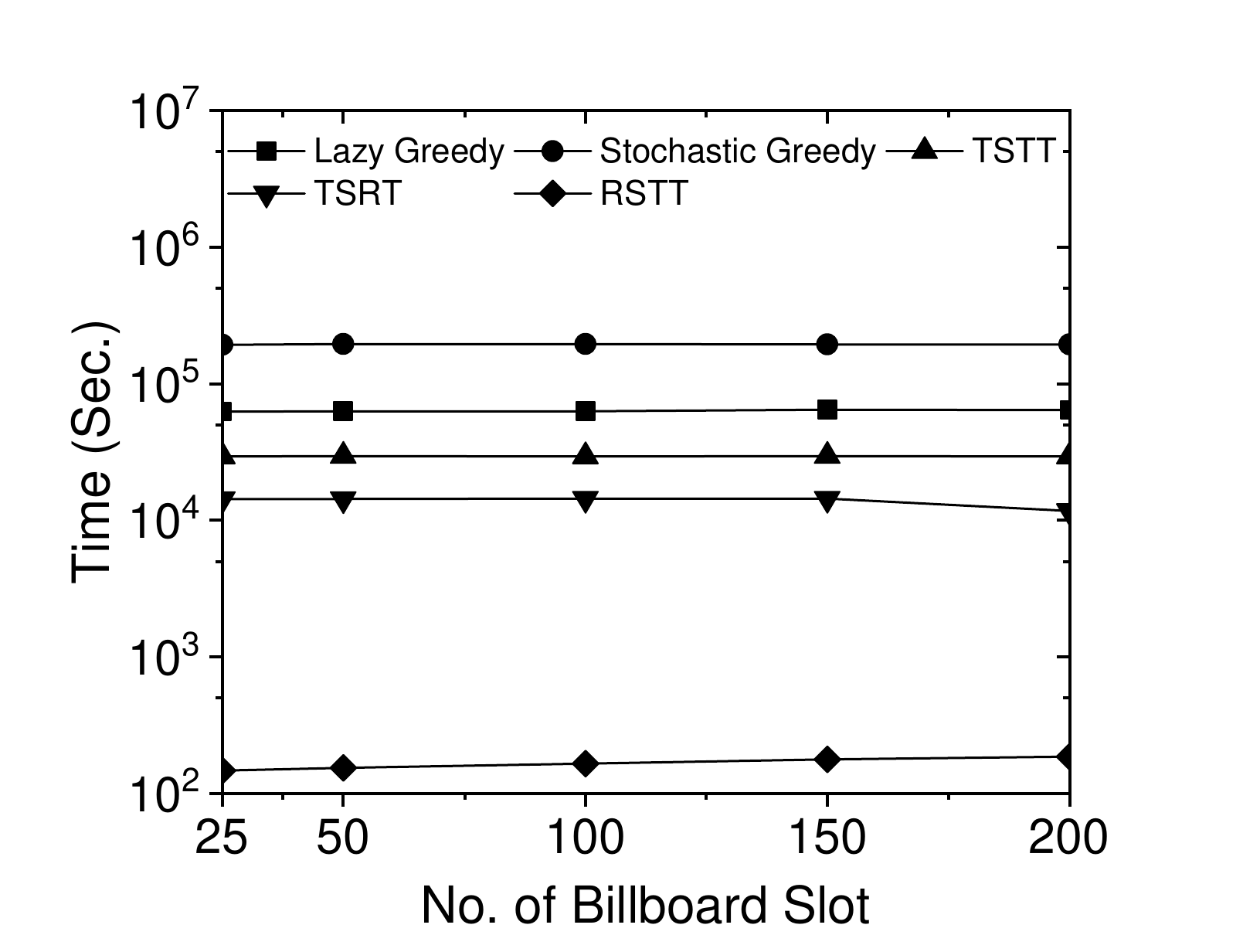}\hspace{-2em} & 
\includegraphics[width=0.29\linewidth]{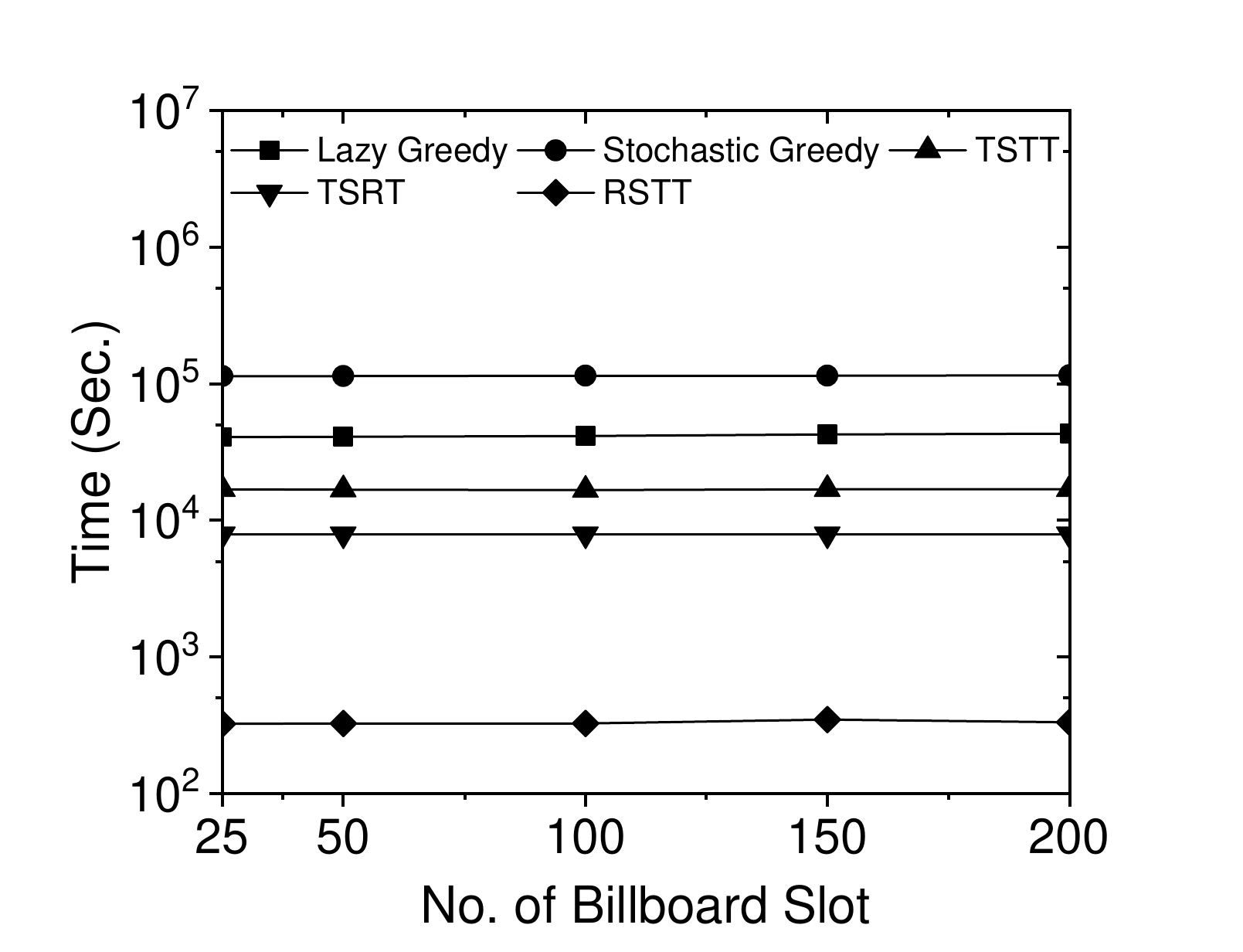} \\
 {\tiny(m) $\ell = 30 $} &  {\tiny(n) $\ell = 40 $} &  {\tiny(o) $\ell = 50$} &  {\tiny(p) $\ell = 10$} \\
\vspace{-0.5em}
\includegraphics[width=0.29\linewidth]{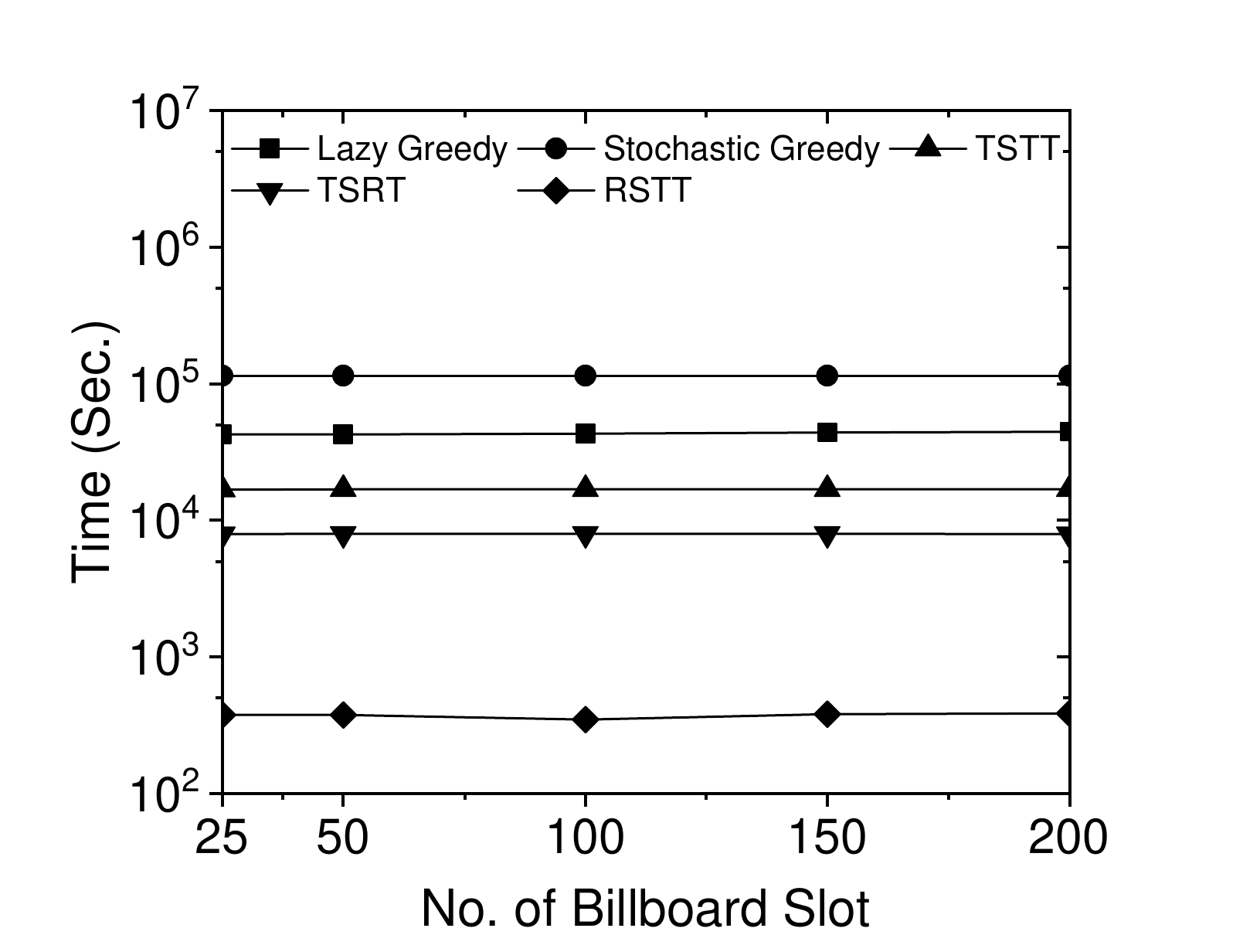}\hspace{-2em} & 
\includegraphics[width=0.29\linewidth]{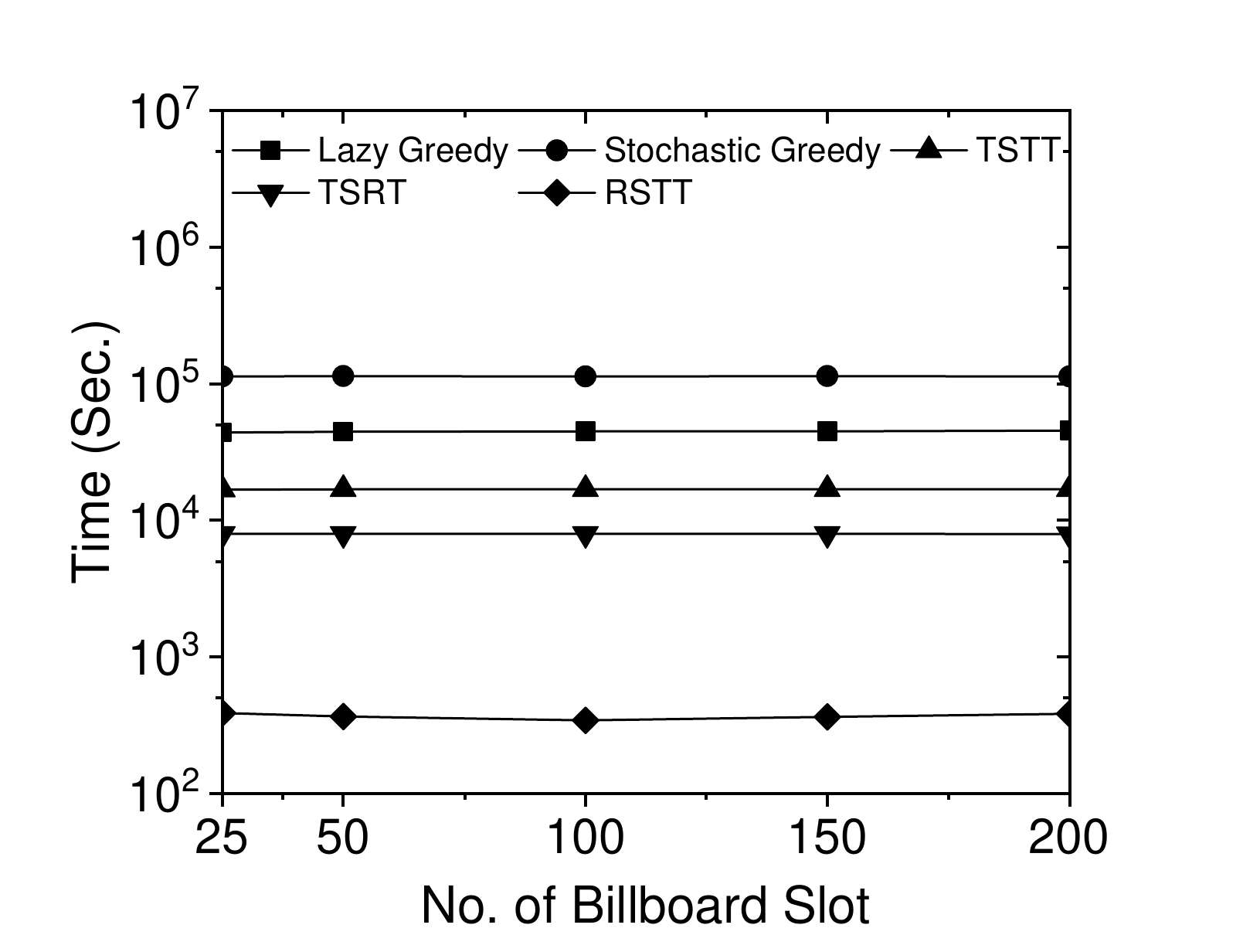}\hspace{-2em} &
\includegraphics[width=0.29\linewidth]{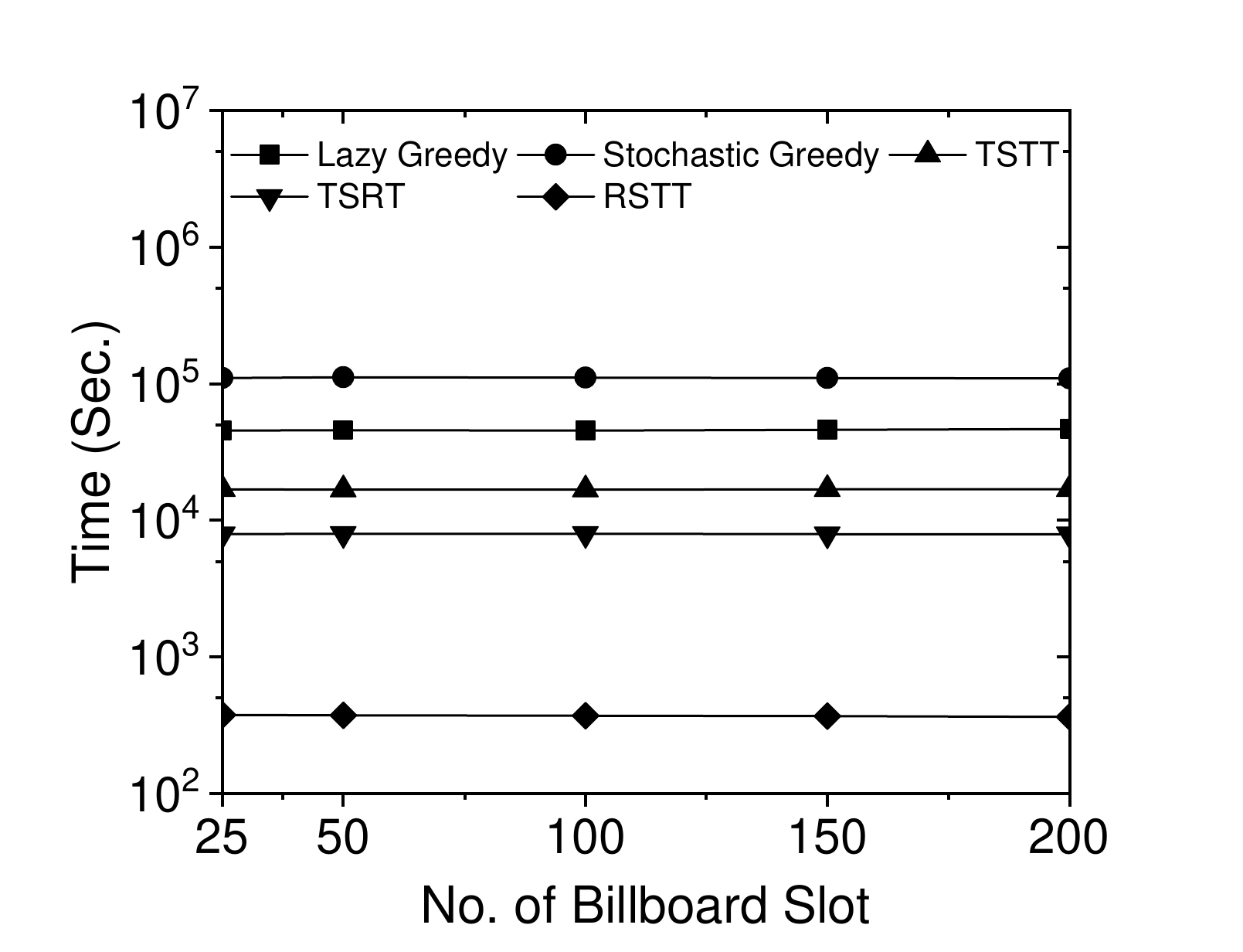}\hspace{-2em} & 
\includegraphics[width=0.29\linewidth]{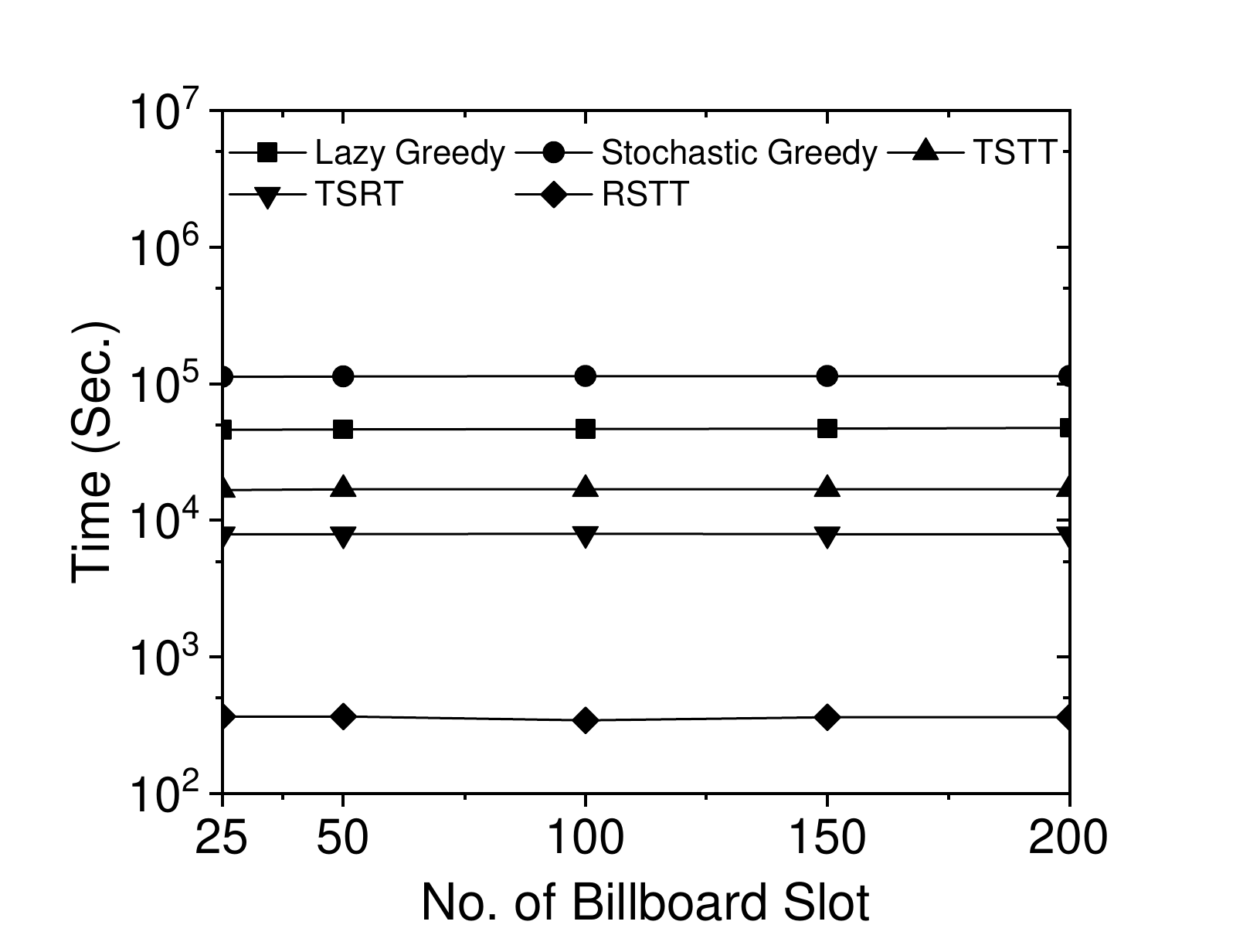} \\
{\tiny (q) $\ell = 20 $} &  {\tiny(r) $\ell = 30 $} &  {\tiny(s) $\ell = 40$} &  {\tiny(t) $\ell = 50$} \\
\end{tabular}
\caption{$(1)$ Influence varying $\ell$, when $k = 25$ to $200$, $\epsilon = 0.01$ ,$(a, b, c, d, e)$ for NYC, and $(f, g, h, i, j)$ for LA Dataset. $(2)$ Time varying $\ell$, when $k = 25$ to $200$, $\epsilon = 0.01$ ,$(k, \ell, m, n, o)$ for NYC, and $(p, q, r, s, t)$ for LA Dataset.}
\label{Fig:Budget_Vs_INF_Time}
\end{figure*}
\vspace{-0.1in}
\paragraph{\textbf{Budget $(k,\ell)$ Vs. Time.}}
To understand the time requirement for proposed and baseline methods, we vary different $k$, and $\ell$ values with respect to time. From Figure \ref{Fig:Budget_Vs_INF_Time}, it is observed that with a fixed value of $\ell$, when $k$ increases, the time requirement also increases. For example, in the LA dataset, when we fixed the value of $\ell=10, \epsilon=0.01$, and varied $k$ value from $25$ to $200$, the time requirement in seconds for `Lazy Greedy', `Stochastic Greedy', `TSTT', `TSRT', and  `RSTT' increases from $40571$, $112259$, $16667$, $7929$, $323$ to $43193$, $114119$, $16885$, $7984$, $332$ respectively. Here, we observed that small changes in time between $k=25$ and $k=200$ happen, and this occurs due to marginal gain computation for each proposed method as well as the baseline method. Similarly, when we set $\ell=50, \epsilon=0.01$, and vary $k=25$ to $k=200$, the time requirement for `Lazy Greedy', `Stochastic Greedy', `TSTT', `TSRT', and  `RSTT' also increases from $46002$, $113872$, $16855$, $7945$, $360$ to $47652$, $115215$, $16907$, $7993$, $365$ respectively. One point needs to be noted that the experimental results of `Lazy Greedy' are reported in Figure \ref{Fig:Budget_Vs_INF_Time}, which is the best case time requirements, and in the worst case, it will take the same run time as `Incremental Greedy'. However, when the dataset is large, `Lazy Greedy' may not be the right choice. Now, in the case of `Stochastic Greedy', its computational time is always far better than the `Incremental Greedy' method as it is independent of the size of $k$, and $\ell$ as discussed in Lemma \ref{lemma:4}. In the case of the NYC dataset, a similar behavior is observed as of the LA dataset for the proposed and baseline methods. We have not reported the time requirements for the `RSHFT', `MAXSRT', and `RSRT' methods as these methods take less than $10$ seconds of computational time.
\begin{figure*}[h!]
\centering
\begin{tabular}{cccc}
\includegraphics[width=0.29\linewidth]{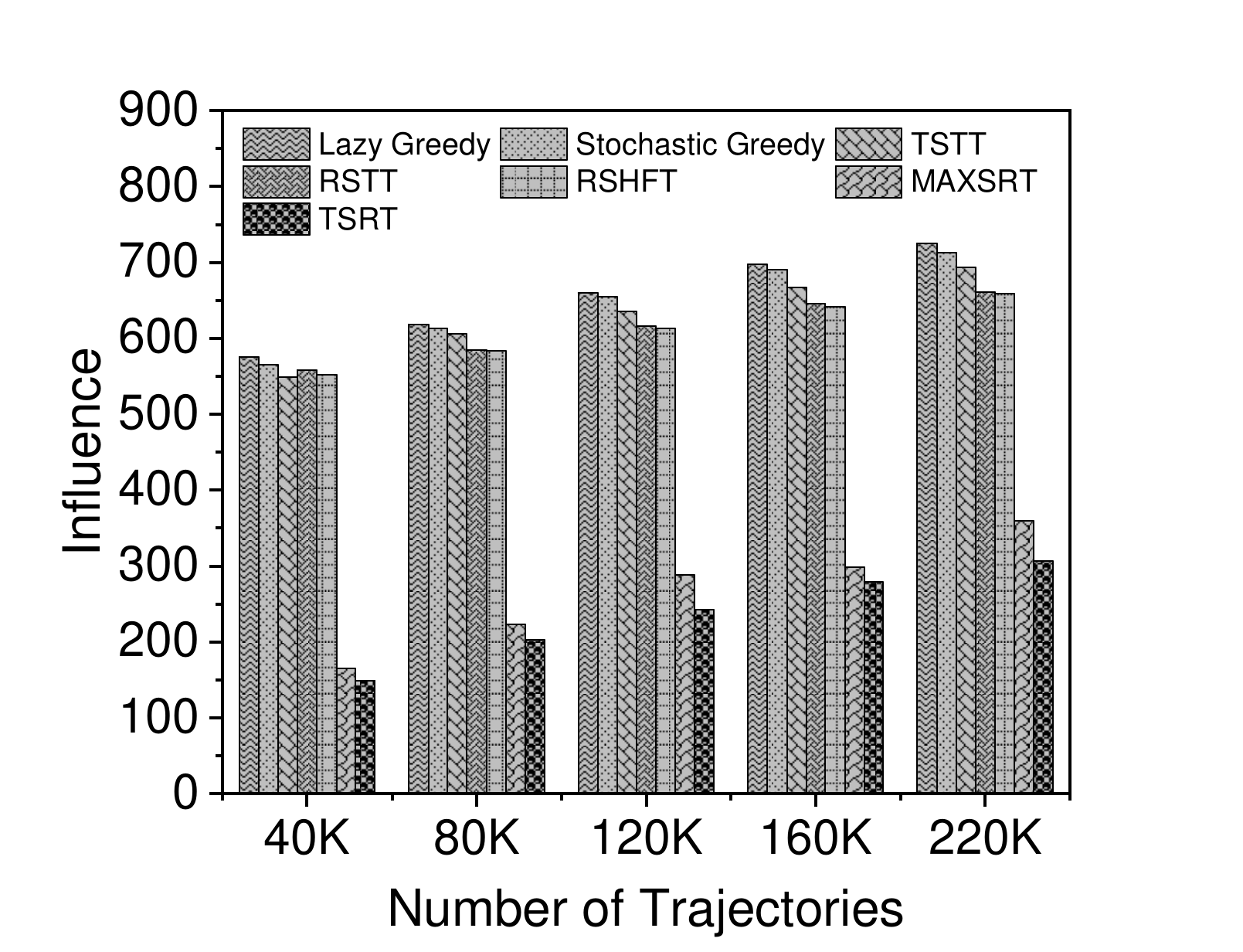}\hspace{-2em} & 
\includegraphics[width=0.29\linewidth]{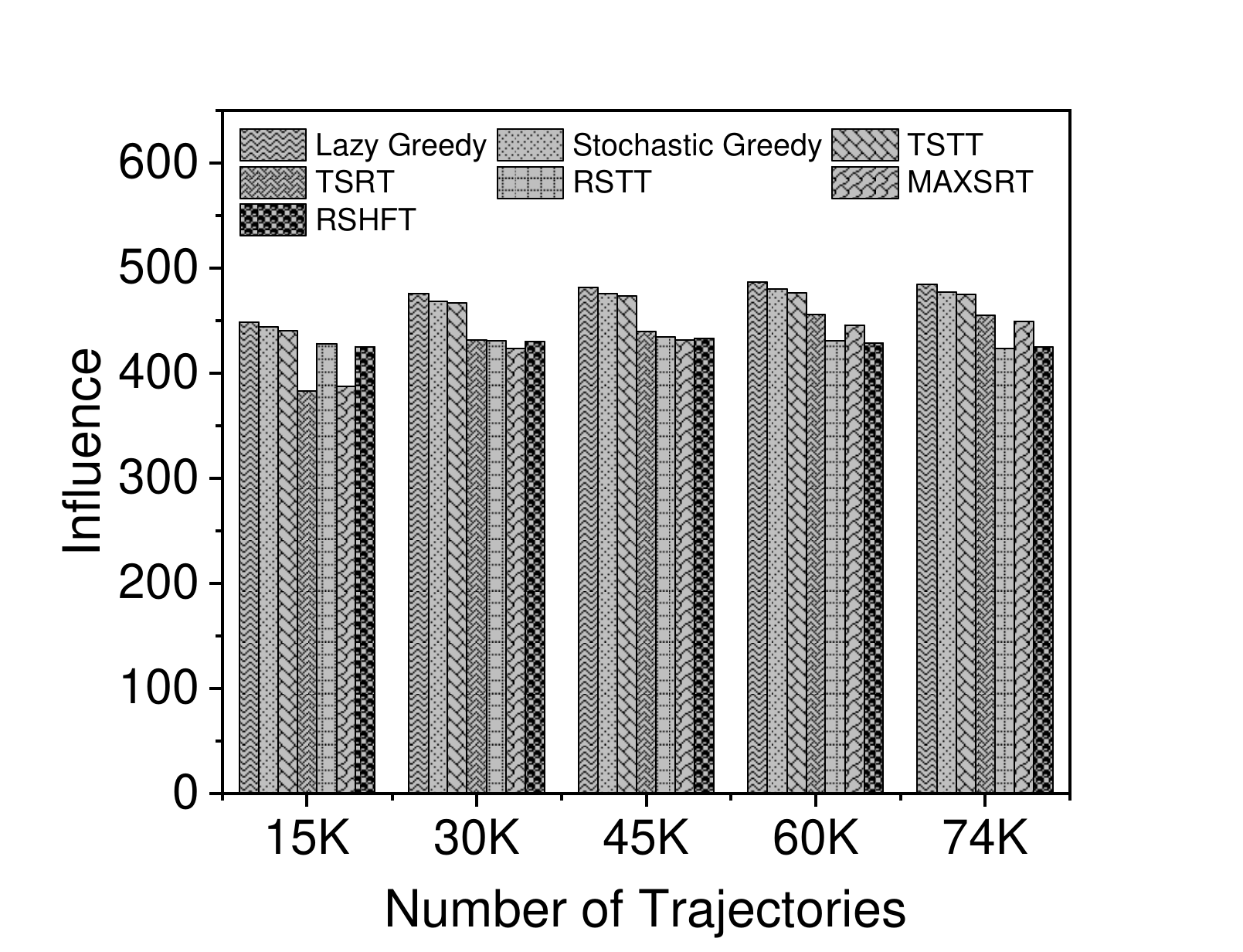}\hspace{-2em}  &
\includegraphics[width=0.29\linewidth]{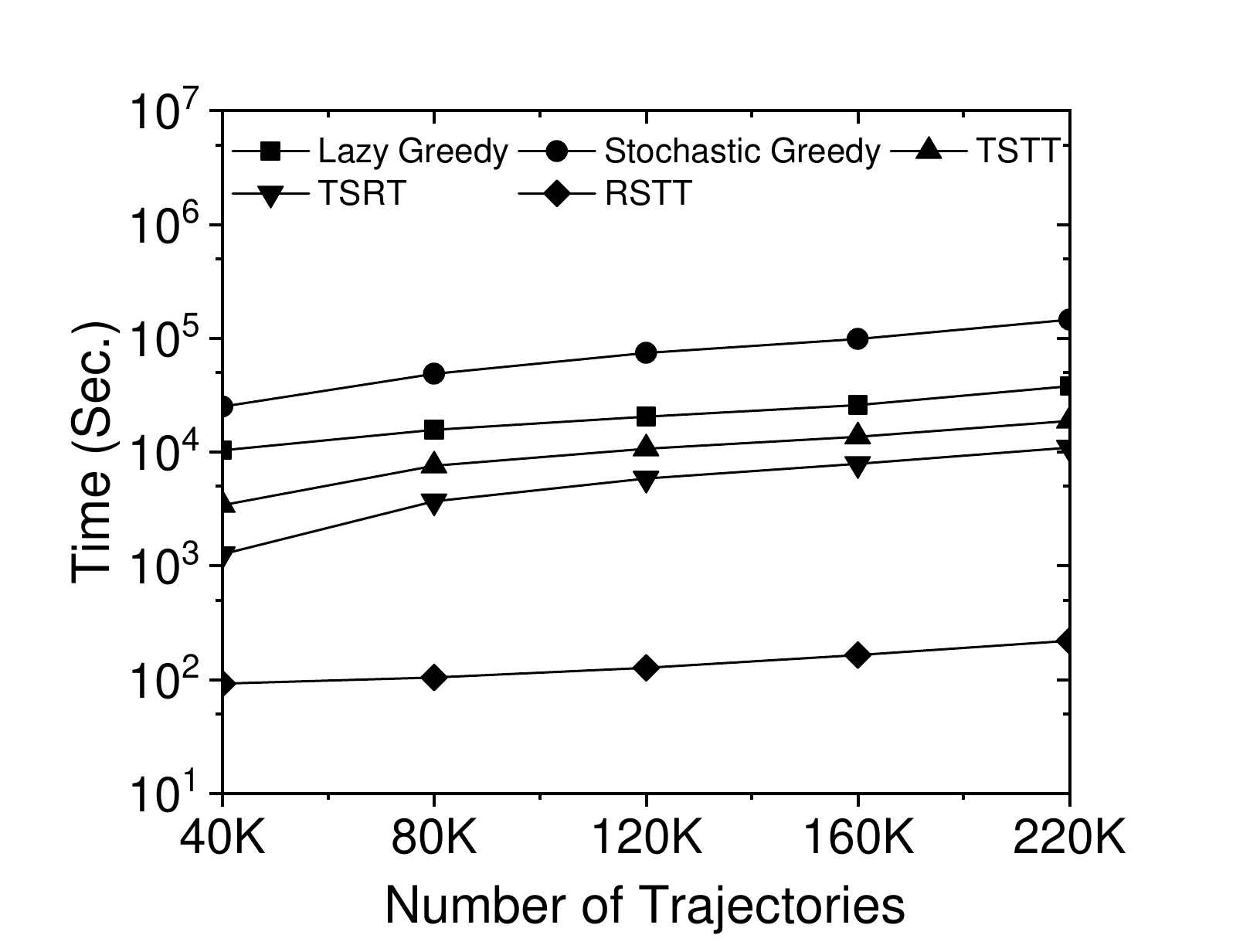}\hspace{-2em} &
\includegraphics[width=0.29\linewidth]{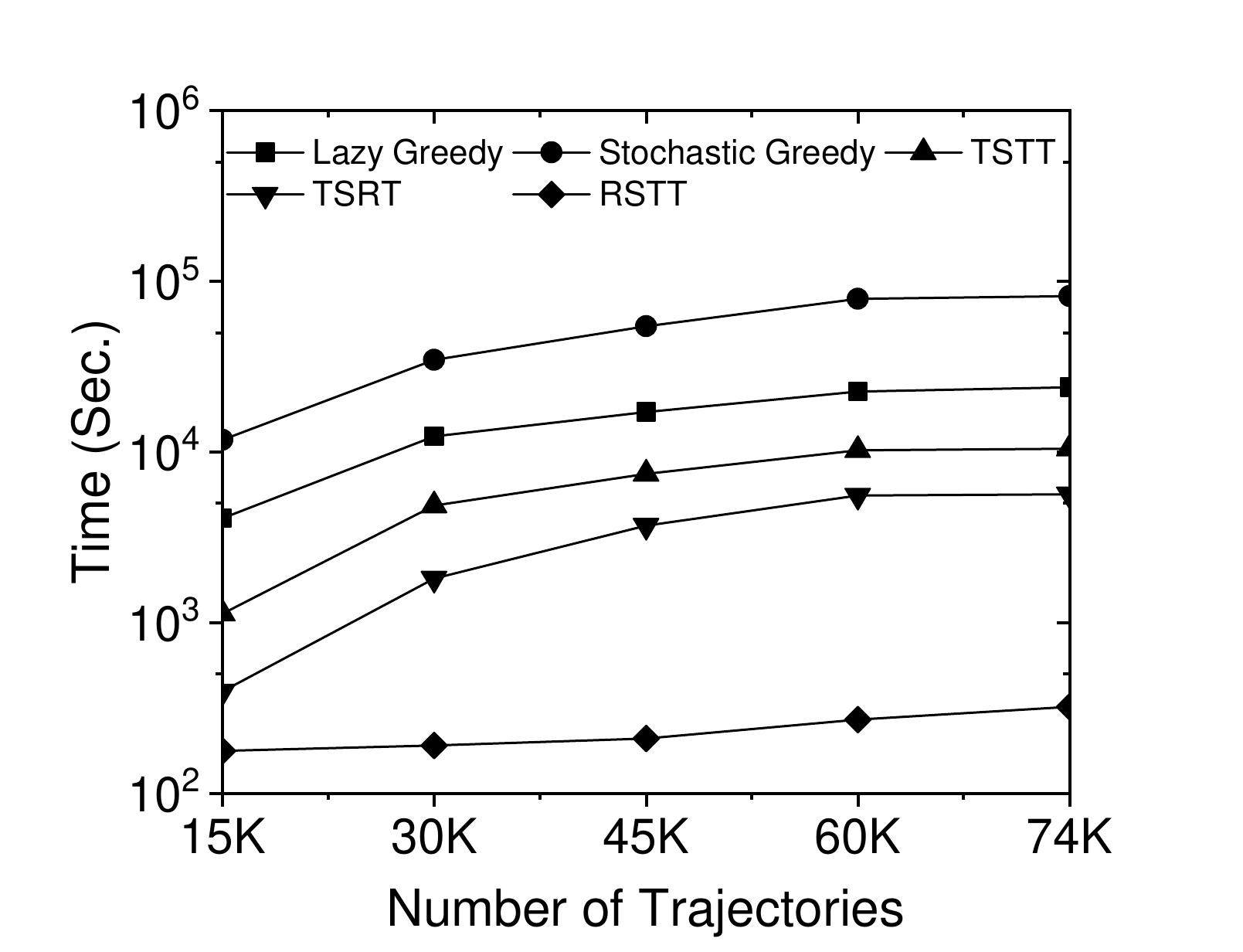} \\
{\tiny(a) $k =200, \ell = 50$} & {\tiny(b) $k =200, \ell = 50$} & {\tiny(c) $k =200, \ell = 50$} & {\tiny(d) $k =200, \ell = 50$}\\
\end{tabular}
\caption{$(1)$ Influence varying trajectory size, when $k = 200$, $\ell = 50$, and $\epsilon = 0.01$ $(a)$ for NYC, and $(b)$ for LA Dataset. $(2)$ Time varying trajectory size, when $k = 200$, $\ell = 50$, and $\epsilon = 0.01$ $(c)$ for NYC, and $(d)$ for LA Dataset.}
\label{Fig:TJ_Vs_INF_Time}
\end{figure*}
\vspace{-0.1in}
\paragraph{\textbf{Trajectory Size Vs. Influence, Time.}}
Figure \ref{Fig:TJ_Vs_INF_Time} shows the impact of varying trajectory size on influence and run time. We observe: (1) the influence of all proposed and baseline methods increases with the increment of trajectory size because more users can be influenced. (2) In the NYC and LA datasets, the influence of `Lazy Greedy', and `Stochastic Greedy' is consistently better than the baseline methods. We take $k=200$, and $\ell =50$, and vary trajectory size $40k$ to $200k$ for the NYC, and $15k$ to $74k$ for the LA dataset as shown in Figure \ref{Fig:TJ_Vs_INF_Time}(a), \ref{Fig:TJ_Vs_INF_Time}(b). (3) In the NYC dataset, when trajectory sizes are $40k$, $80k$, $120k$, $160k$, $200k$, and their corresponding unique users encountered are $924$, $969$, $1017$, $1064$, $1083$, respectively. In the LA dataset, when the trajectory varies between $15k$ to $74k$, the number of unique users encountered is $2000$. (4) Figures \ref{Fig:TJ_Vs_INF_Time}(c), and \ref{Fig:TJ_Vs_INF_Time}(d) shows computational time for the NYC and LA dataset. We observe that `Lazy Greedy', and `Stochastic Greedy' scale linearly w.r.t. trajectory size, consistent in our analysis in both the NYC and LA datasets.  Although the growth in time requirement in `Stochastic Greedy' is faster than `Lazy Greedy', e.g., when trajectory size varies from $40k$ to $220k$, and $15k$ to $74k$, the time requirement increases almost $6\times$ and $6.5\times$ times for NYC and LA datasets, respectively. However, in the `Lazy Greedy', run time rises linearly in the best case as only one time marginal gain needs to be computed for all the elements, and from on-wards only comparison operation needs to be executed. (5) Among the baseline methods, `TSTT' takes the maximum time, and with the increase of trajectory size, the run time of all baseline methods increases linearly.
% Plot for LA Dataset
%%%%%%%%%%%%%%%%%%%%%%%%%%%%%%%%%%%%%%%%%%%%%%%%%%%%%%%%%%%%%%%%
\begin{figure*}[h!]
\centering
\begin{tabular}{cccc}
\includegraphics[width=0.29\linewidth]{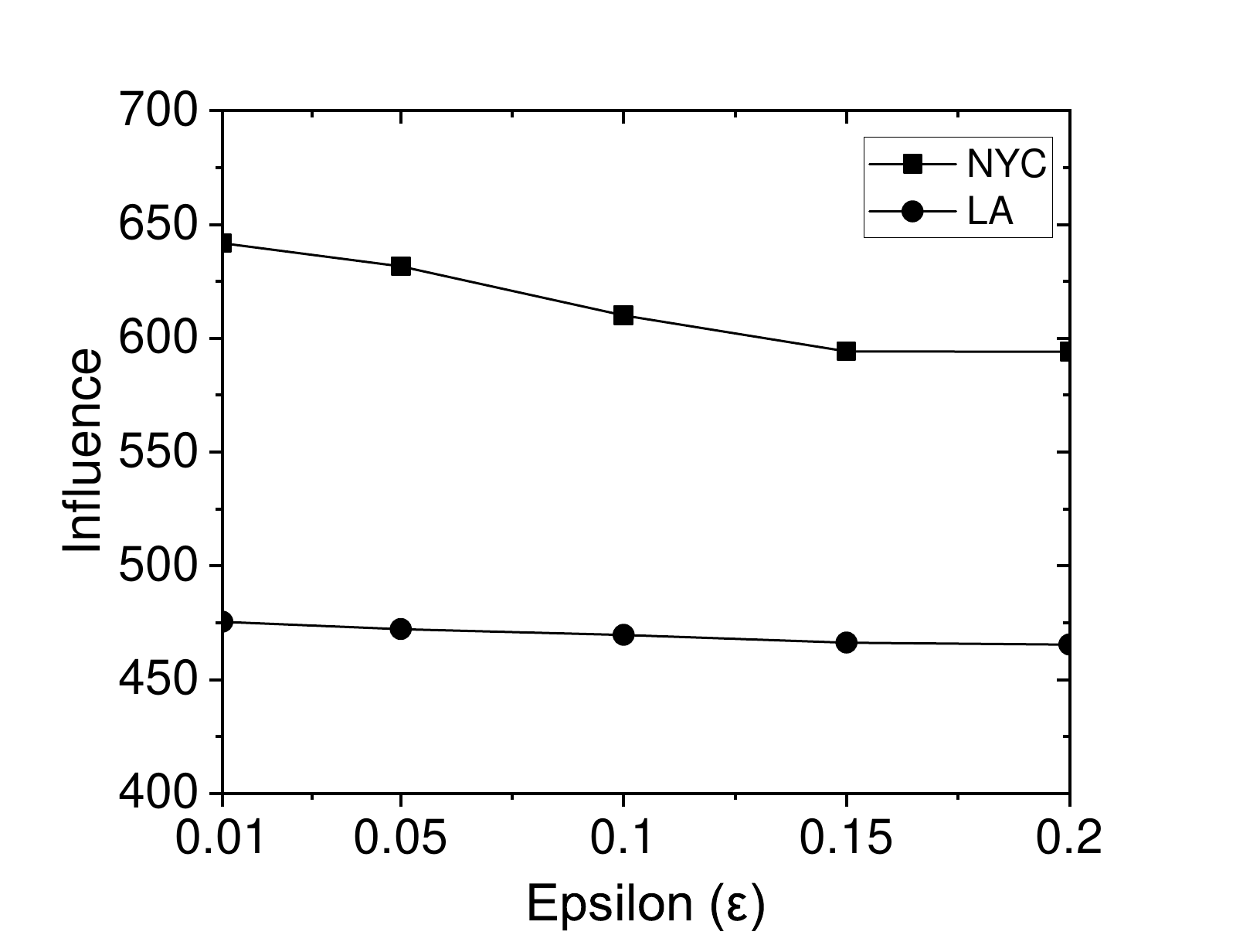}\hspace{-2em} & 
\includegraphics[width=0.29\linewidth]{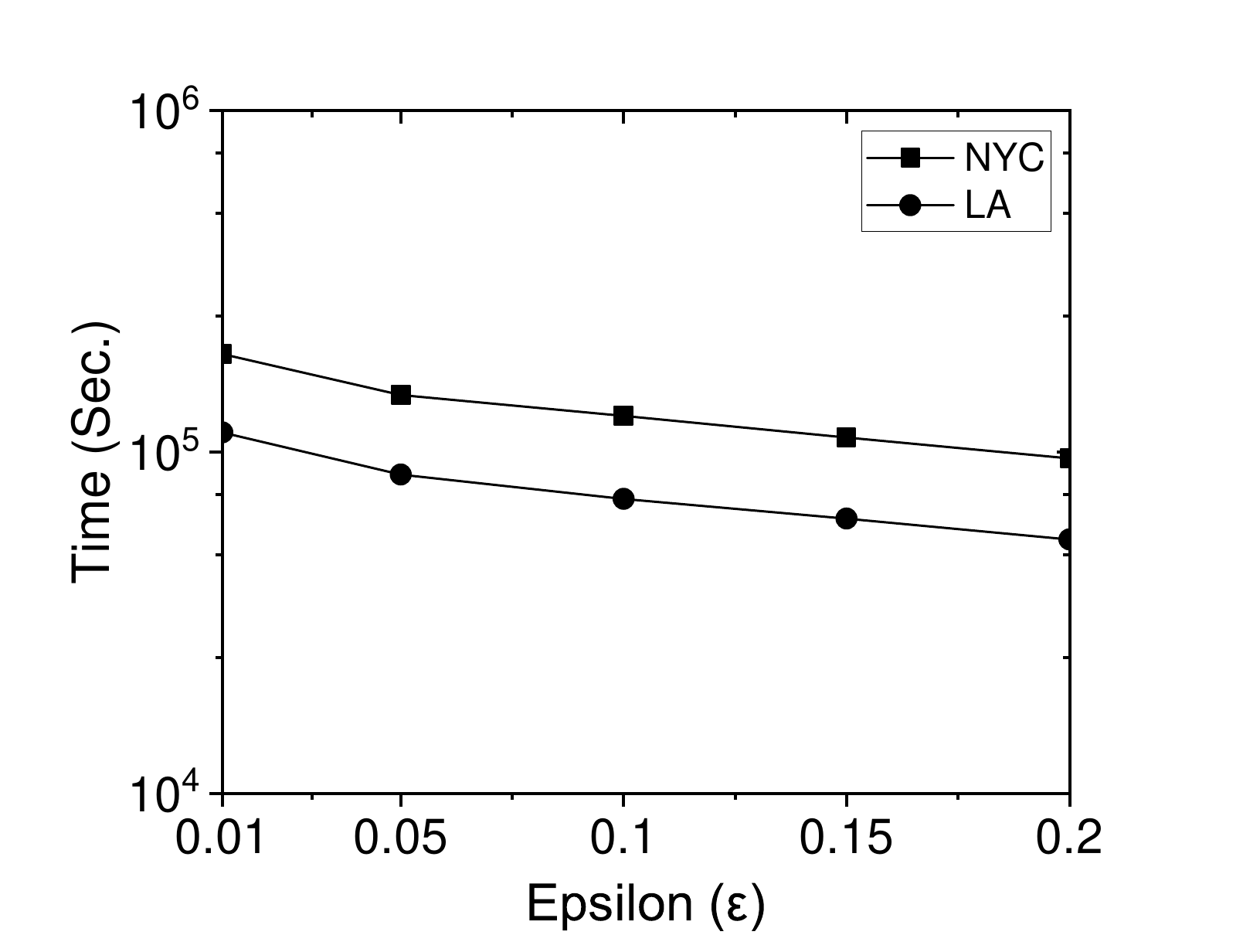}\hspace{-2em} &
\includegraphics[width=0.29\linewidth]{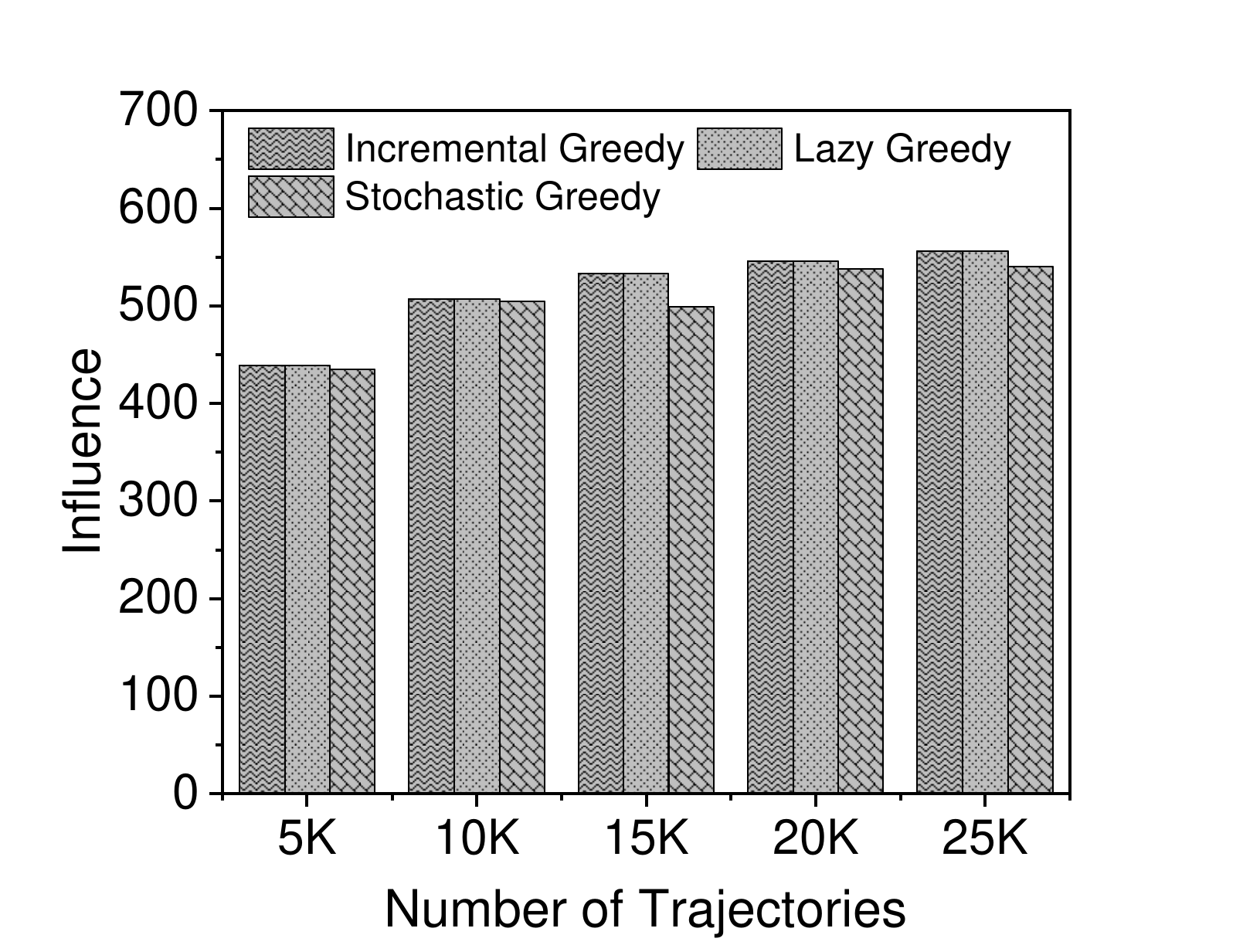}\hspace{-2em} & \includegraphics[width=0.29\linewidth]{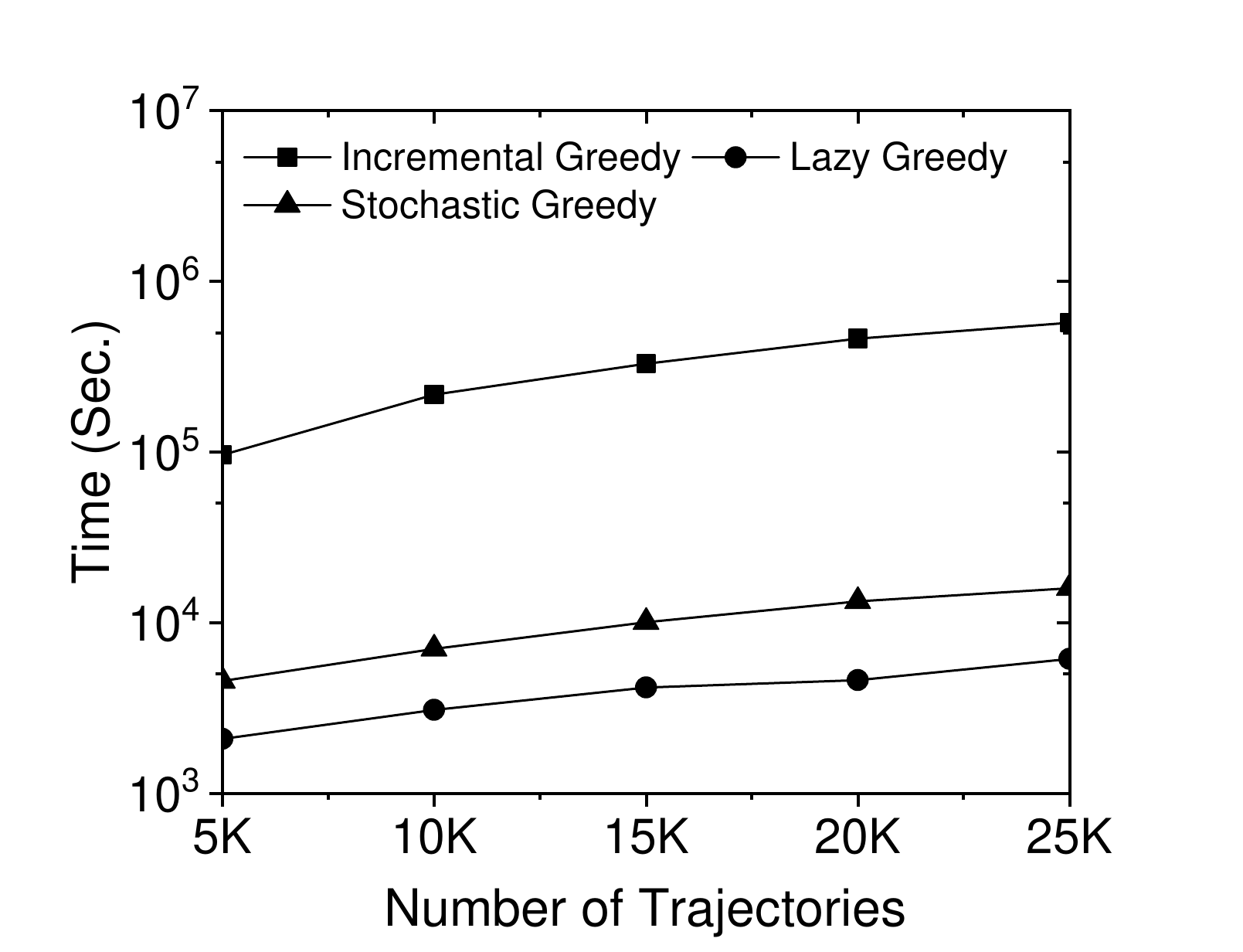} \\
{\tiny(a) $k =200, \ell = 50$} &  {\tiny(b) $k =200, \ell = 50 $} &  {\tiny(c) $k =200, \ell = 50$} & {\tiny (d) $k =200, \ell = 50$}\\
\end{tabular}
\caption{$(1)$ Influence varying $\epsilon ~(a)$, and Time Varying $\epsilon ~(b)$ when $k = 200$, $\ell = 50$ for NYC, LA Dataset. $(2)$ Influence $(c)$, Time $(d)$, varying trajectory size, when $k = 200$, $\ell = 50$, and $\epsilon = 0.01$ for different Algorithms on NYC Dataset.}
\label{Fig:Epsilon_vs_Time_inf_greedy}
\end{figure*}
\vspace{-0.2in}
%%%%%%%%%%%%%%%%%%%%%%%%%%%%%%%%%%%%%%%%%%%%%%%%%%%%%%%%%%%%%%%%%%
\paragraph{\textbf{Epsilon ($\epsilon$) Vs. Influence, Time.}}
Figure \ref{Fig:Epsilon_vs_Time_inf_greedy}(a), \ref{Fig:Epsilon_vs_Time_inf_greedy}(b) shows the impact of varying $\epsilon$ values on `Stochastic Greedy' w.r.t. influence, and time. We find: (1) when the $\epsilon$ value increases, the influence value decreases. The influence is decreasing more in the NYC dataset than in the LA dataset. (2) When the $\epsilon$ value varies from $0.01$ to $0.2$, the run time on both the NYC and LA datasets decreases linearly. In the `Stochastic Greedy', we randomly pick a subset of elements, and the cardinality of the subset depends on the $\epsilon$ value. If the $\epsilon$ value decreases, then the subset size increases, and there is a minimal loss in influence compared to `Incremental Greedy', however run time increases. For example, in the NYC dataset, when $k = 200,\ell = 50$ and vary $\epsilon$ for the value of $0.01$ to $0.2$, the influence values are $641.69$, $631.60$, $610.11$, $594.27$, $593.12$ and the run-times are $193850$, $147011$, $127458$, $110355$, $95880$ in seconds, respectively. A similar type of result was observed on the LA dataset as shown in Figure \ref{Fig:Epsilon_vs_Time_inf_greedy}(a), \ref{Fig:Epsilon_vs_Time_inf_greedy}(b). So, the parameter $\epsilon$, gives us the freedom to compromise either in influence or run time.
\vspace{-0.2in}
%%%%%%%%%%%%%%%%%%%%%%%%%%%%%%%%%%%%%%%%%%%%%%%%%%%%%%%%%%%%%%%%
\begin{figure*}[h!]
\centering
\begin{tabular}{cccc}
\hspace{-2em}
\includegraphics[width=0.29\textwidth]{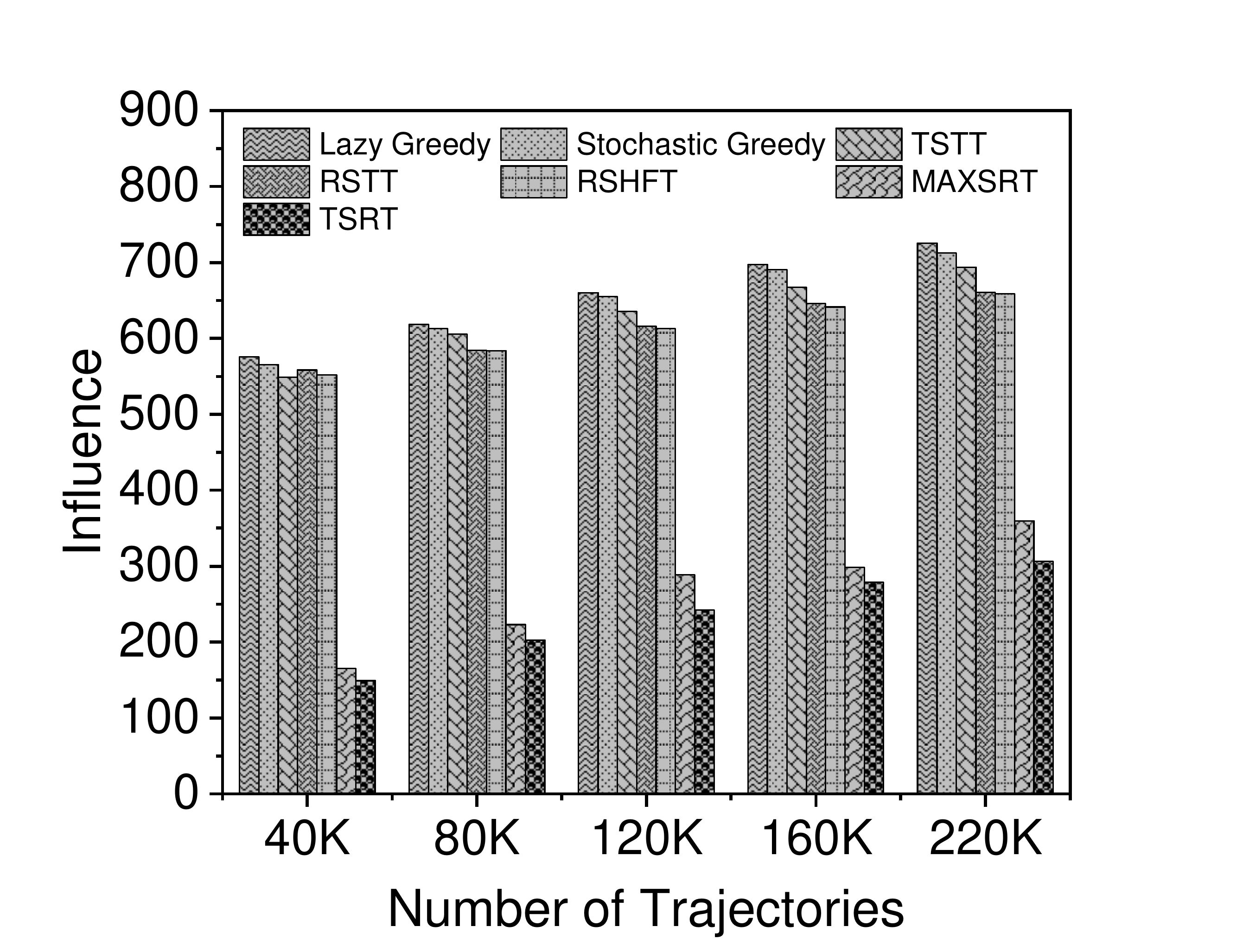} &
\hspace{-2em}
\includegraphics[width=0.29\textwidth]{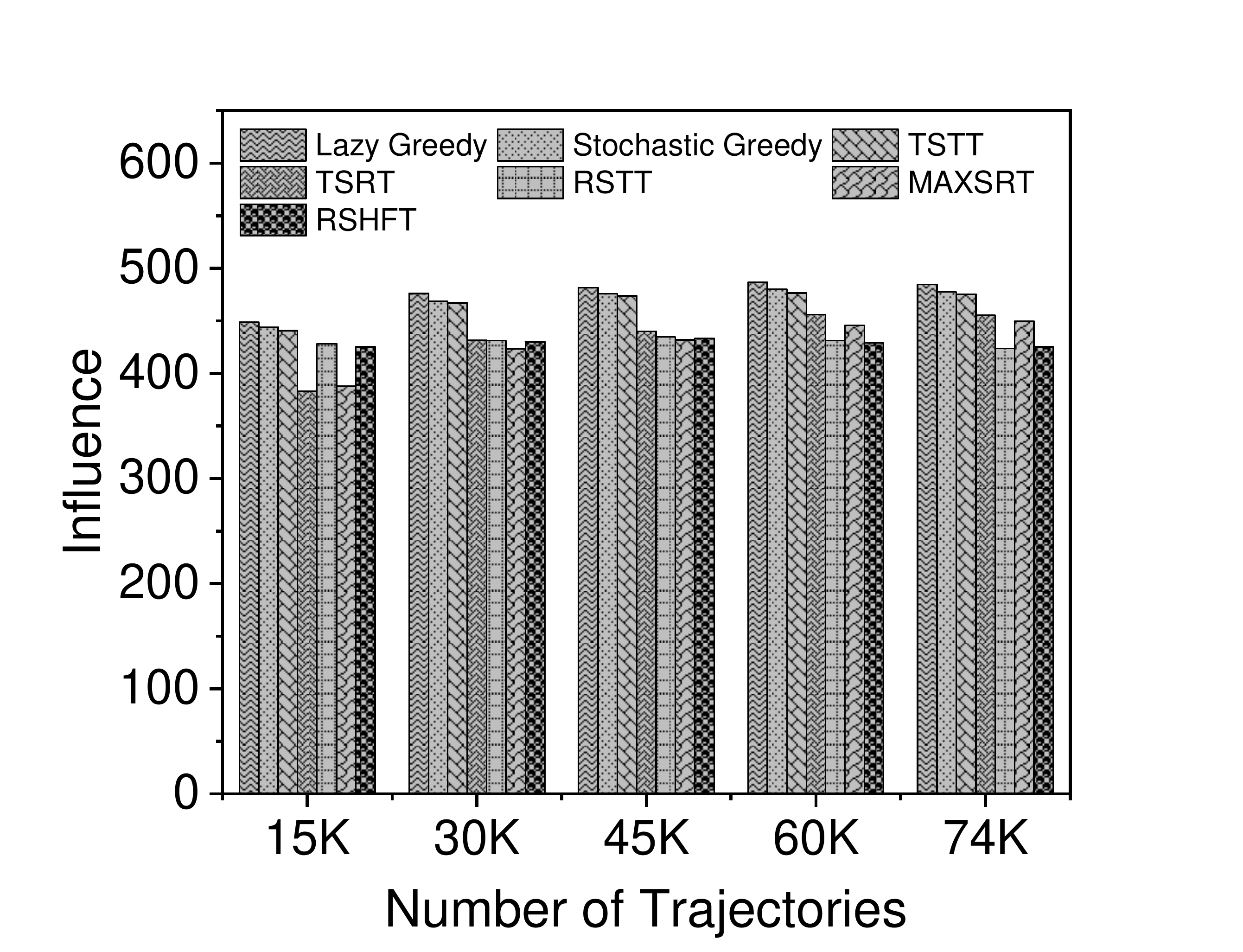} &
\hspace{-2em}
\includegraphics[width=0.29\textwidth]{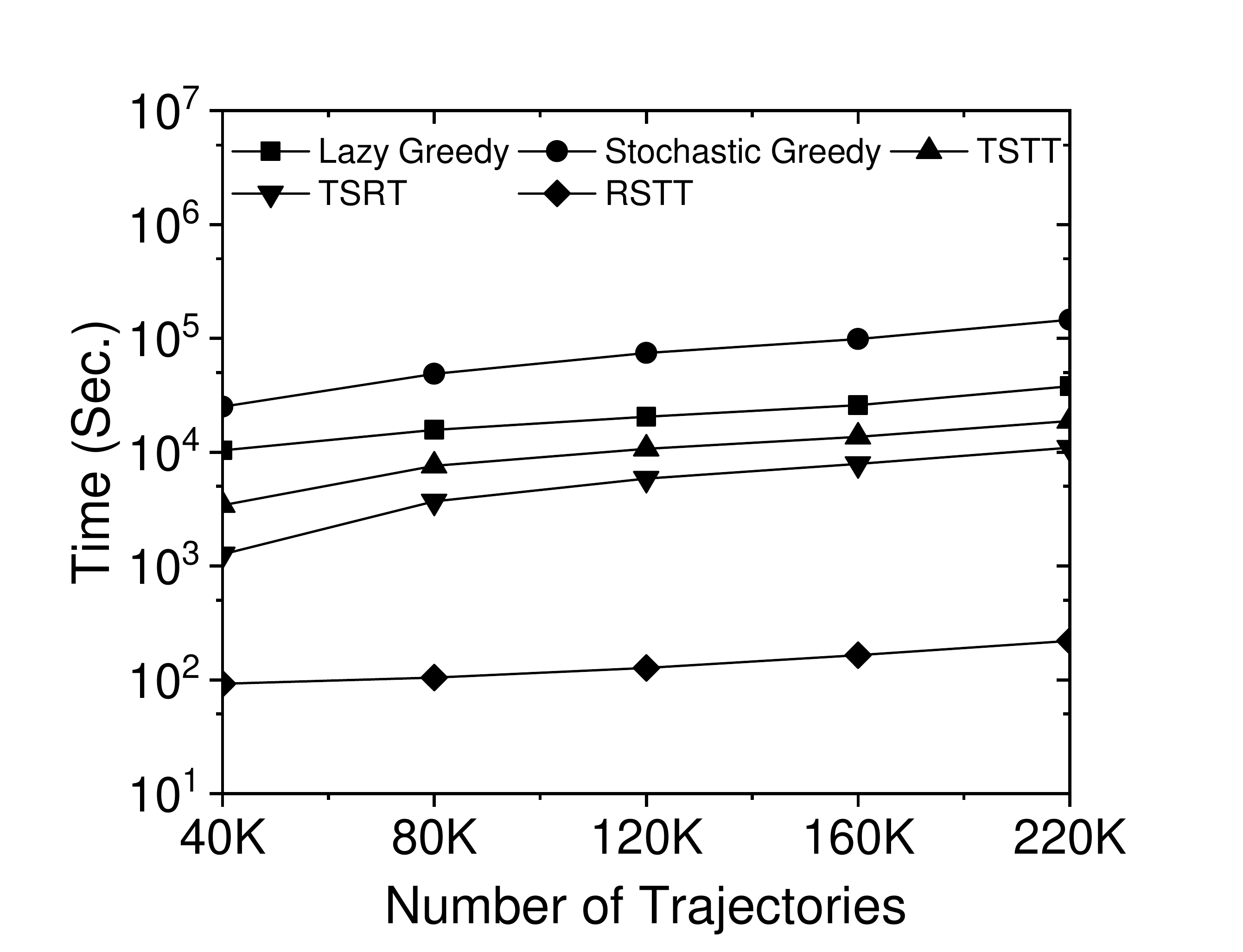} &
\hspace{-2em}
\includegraphics[width=0.29\textwidth]{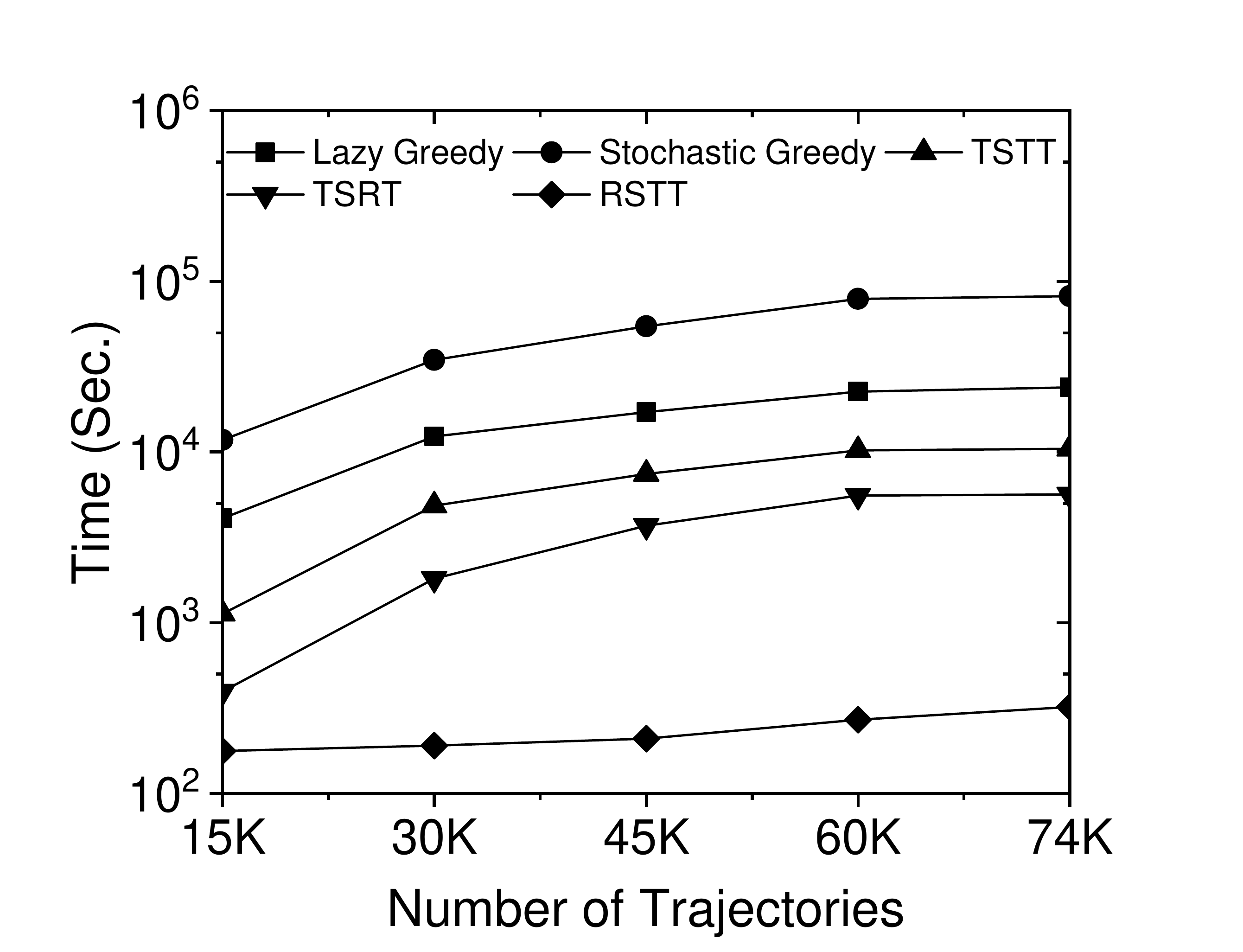} \\
{\tiny(a) $k=200,\ell=50$} &  
{\tiny(b) $k=200,\ell=50$} &  
{\tiny(c) $k=200,\ell=50$} & 
{\tiny(d) $k=200,\ell=50$}\\
\end{tabular}
\caption{(1) Influence varying Distance $(\lambda)$ when $k=200$, $\ell=50$, and $\epsilon=0.01$: (a) NYC Dataset, (b) LA Dataset. (2) Time varying Distance $(\lambda)$ when $k=200$, $\ell=50$, and $\epsilon=0.01$: (c) NYC Dataset, (d) LA Dataset.}
\label{Fig:Distance_TJ_Vs_INF_Time}
\end{figure*}
\vspace{-0.1in}
\paragraph{\textbf{Additional Discussions.}} To find out the efficiency of `Stochastic Greedy', we compare its performance with `Incremental Greedy' and `Lazy Greedy'. In our experiment, we fixed $k, \ell, \epsilon$ value, which varies over different trajectory sizes. Our experiment shows that `Incremental Greedy' and `Lazy Greedy' achieve the same amount of influence; however, there is a huge difference when talking about run time. The `Stochastic Greedy' achieves less influence than both `Incremental Greedy' and `Lazy Greedy' however it takes much less run time compared to `Incremental Greedy'. As we previously discussed, in the worst case, `Lazy Greedy' will take the same amount of time as `Incremental Greedy' takes, and in our experiment, we, fortunately, got the best case results of `Lazy Greedy' due to the nature of the datasets, as reported in Figure \ref{Fig:Epsilon_vs_Time_inf_greedy}(c), \ref{Fig:Epsilon_vs_Time_inf_greedy}(d). When trajectory size increases from $5k$ to $25k$, the run time of `Incremental Greedy', `Lazy Greedy', and `Stochastic Greedy' also increases from $96044$, $2090$, $4546$ to $572512$, $6144$, $15905$ seconds, i.e., $6$x, $3$x, $3.5$x respectively. So, for trajectory size $25k$, `Incremental Greedy' will take almost $36$x more time than `Stochastic Greedy', and we observe that for larger trajectory size, i.e., $200k$, the `Incremental Greedy will not complete its execution with a reasonable computational time. We take $\lambda = 100~meter$, and assume within the range of $100m$, a billboard slot can influence all trajectories with a certain probability as shown in Figure \ref{Fig:Distance_TJ_Vs_INF_Time}. We have also experimented with varying $\lambda$ values from $25m$ to $125m$ and observed that with the increment of $\lambda$ value, the influence as well as run time increases because one billboard slot can influence more number of trajectories.
\vspace{-0.17in}
\section{Conclusion} \label{Sec:Conclusion}
This paper has studied the problem of jointly selecting influential billboard slots and influential tags. First, we show that the influence function is non-negative, monotone, and bi-submodular. We show that the problem is \textsf{NP-hard} and propose an orthant-wise incremental greedy algorithm that gives a constant factor approximation algorithm. Though this method is simple to understand, it does not scale well when the trajectory dataset is large due to excessive marginal gain computations. To address this, we propose the orthant-wise Lazy and Stochastic Greedy approach, which executes fast while leading to more or less similar influence. Still, the problem is not solved on the ground because we must also report which tag will be displayed in which slot to maximize the influence. Developing more efficient techniques to address slot selection and allocation problems will remain an active area of research in the near future. 
\bibliographystyle{splncs04}
\bibliography{Paper}
\end{document}